\documentclass[a4paper,USenglish,cleveref,autoref,thm-restate]{article}
\usepackage{fullpage}

\usepackage{graphicx}
\usepackage{amssymb}
\usepackage{caption}
\usepackage{subcaption}
\usepackage{multirow}
\usepackage{url}
\usepackage{amsmath}
\usepackage{amsthm}
\usepackage{mathtools}
\usepackage{algorithm}
\usepackage{algpseudocode}
\usepackage{listings}
\usepackage{comment}

\newtheorem{theorem}{Theorem}[section]
\newtheorem{lemma}[theorem]{Lemma}

\newtheorem{observation}[theorem]{Observation}
\newtheorem{claim}[theorem]{Claim}
\newtheorem{subclaim}[theorem]{Sub-Claim}

\newenvironment{definition}[1][Definition]{\begin{trivlist}
		\item[\hskip \labelsep {\bfseries #1}]}{\end{trivlist}}

\newcommand{\MYQUEUE}{Jiffy}

\title{\MYQUEUE: A Fast, Memory Efficient, Wait-Free Multi-Producers Single-Consumer Queue}

\author{Dolev Adas \ \ \ and \ \ \  Roy Friedman\\
	Computer Science Department\\
	Technion\\
	\texttt{\{sdolevfe,roy\}@cs.technion.ac.il}
}

\begin{document}
	
\maketitle

\begin{abstract}
In applications such as sharded data processing systems, sharded in-memory key-value stores, data flow programming and load sharing applications, multiple concurrent data producers are feeding requests into the same data consumer.
This can be naturally realized through concurrent queues, where each consumer pulls its tasks from its dedicated queue.
For scalability, wait-free queues are often preferred over lock based structures.

The vast majority of wait-free queue implementations, and even lock-free ones, support the multi-producer multi-consumer model.
Yet, this comes at a premium, since implementing wait-free multi-producer multi-consumer queues requires utilizing complex helper data structures.
The latter increases the memory consumption of such queues and limits their performance and scalability.
Additionally, many such designs employ (hardware) cache unfriendly memory access patterns.

In this work we study the implementation of wait-free multi-producer single-consumer queues.
Specifically, we propose \MYQUEUE, an efficient memory frugal novel wait-free multi-producer single-consumer queue and formally prove its correctness.
We then compare the performance and memory requirements of \MYQUEUE{} with other state of the art lock-free and wait-free queues.
We show that indeed \MYQUEUE{} can maintain good performance with up to 128 threads, delivers up to $50\%$ better throughput than the next best construction we compared against, and consumes $\approx$90\% less memory.
\end{abstract}
\ifdefined\COVER
\newpage
\fi
\section{Introduction}
\label{sec:intro}

Concurrent queues are a fundamental data-exchange mechanism in multi-threaded applications.
A queue enables one thread to pass a data item to another thread in a decoupled manner, while preserving ordering between operations.
The thread inserting a data item is often referred to as the \emph{producer} or \emph{enqueuer} of the data, while the thread that fetches and removes the data item from the queue is often referred to as the \emph{consumer} or \emph{dequeuer} of the data.
In particular, queues can be used to pass data from multiple threads to a single thread - known as \emph{multi-producer single-consumer queue} (MPSC), from a single thread to multiple threads - known as \emph{single-producer multi-consumer queue} (SPMC), or from multiple threads to multiple threads - known as \emph{multi-producer multi-consumer queue} (MPMC).
\ifdefined\COVER
\else
SPMC and MPSC queues are demonstrated in Figure~\ref{fig:queues}.
\begin{figure}[t]
	\begin{subfigure}{0.36\textwidth}
		\center{
			\includegraphics[width=0.4\textwidth]{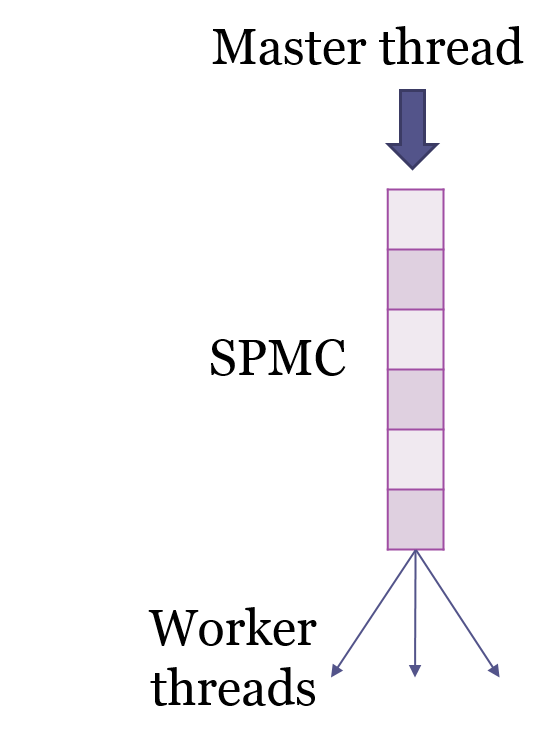}
		}
		\caption{SPMC queue used in a master worker architecture. A single master queues up tasks to be executed, which are picked by worker threads on a first comes first served basis.}
		\label{fig:SPMC}
	\end{subfigure}
\ \ 
	\begin{subfigure}{0.60\textwidth}
		\center{
			\includegraphics[width=0.6\textwidth]{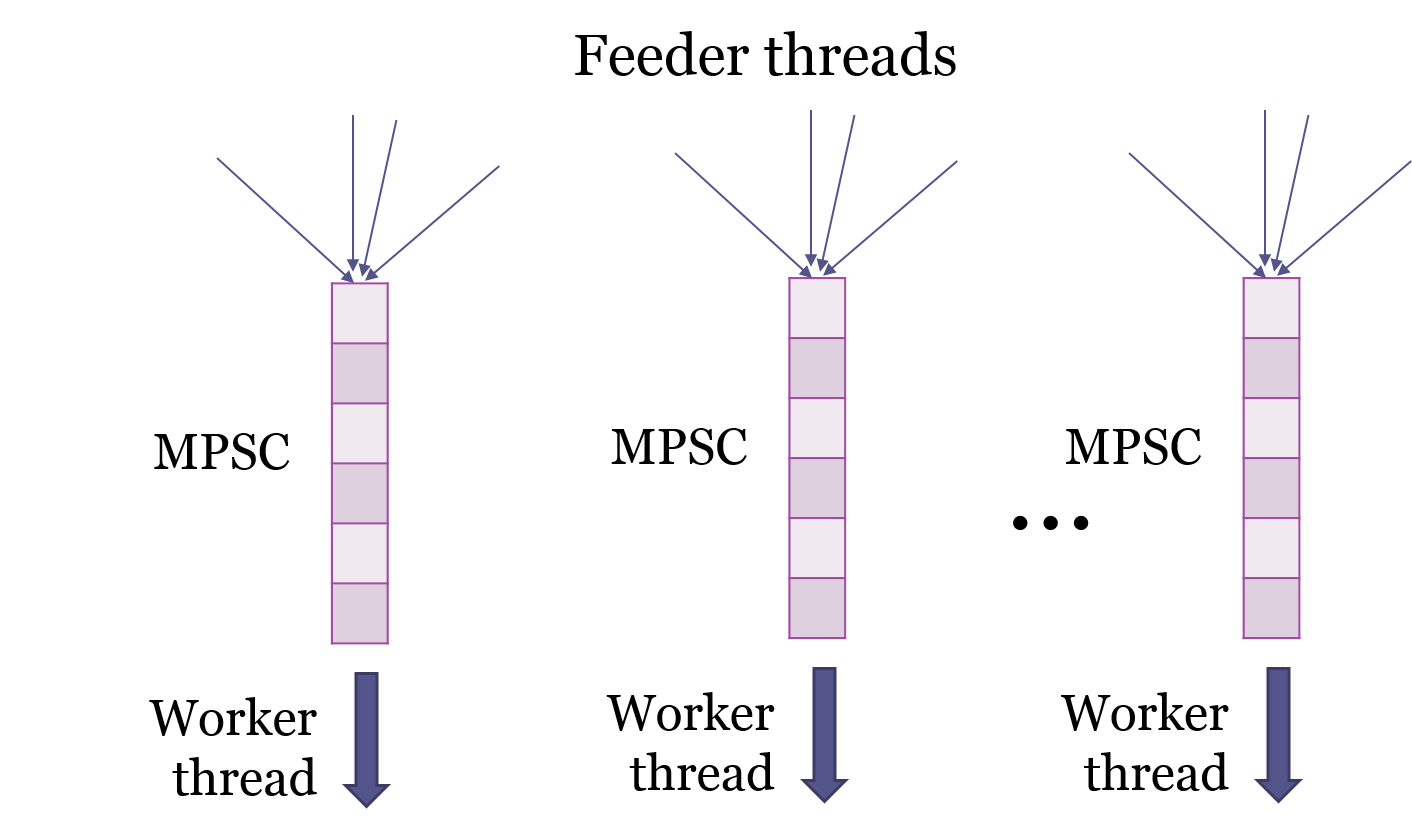}	
		}
		\caption{MPSC queues used in a sharded architecture. Here, each shard is served by a single worker thread to avoid synchronization inside the shard. Multiple collector threads can feed the queue of each~shard.}
		\label{fig:MPSC}
	\end{subfigure}
	\caption{SPMC vs. MPSC queues.}
	\label{fig:queues}
\end{figure}
\fi

MPSC is useful in sharded software architectures, and in particular for sharded in-memory key-value stores and sharded in-memory databases, resource allocation 
and data-flow computation schemes.
Another example is the popular Caffeine Java caching library~\cite{CaffeineProject}, in which a single thread is responsible for updating the internal cache data structures and meta-data.
As depicted in Figure~\ref{fig:MPSC}, in such architectures, a single thread is responsible for each shard, in order to avoid costly synchronization while manipulating the state of a specific shard.
In this case, multiple feeder threads (e.g., that communicate with external clients) insert requests into the queues according to the shards.
Each thread that is responsible for a given shard then repeatedly dequeues the next request for the shard, executes it, dequeues the next request, etc.
Similarly, in a data flow graph, multiple events may feed the same computational entity (e.g., a reducer that reduces the outcome of multiple mappers).
Here, again, each computational entity might be served by a single thread while multiple threads are sending it items, or requests, to be handled.

MPMC is the most general form of a queue and can be used in any scenario.
Therefore, MPMC is also the most widely studied data structure~\cite{DBLP:conf/wdag/Ladan-MozesS04,le2013correct,MS96,SLS06,DBLP:conf/podc/Scott02,DBLP:conf/ppopp/ScottS01,DBLP:conf/podc/ShavitZ99}.
Yet, this may come at a premium compared to using a more specific queue implementation.

Specifically, concurrent accesses to the same data structure require adequate concurrency control to ensure correctness.
The simplest option is to lock the entire data structure on each access, but this usually dramatically reduces performance due to the sequentiality and contention it imposes~\cite{herlihy2011art}.
A more promising approach is to reduce, or even avoid, the use of locks and replace them with \emph{lock-free} and \emph{wait-free} protocols that only rely on atomic operations such as \emph{fetch-and-add} (FAA) and \emph{compare-and-swap} (CAS), which are supported in most modern hardware architectures~\cite{herlihy1991wait}.
Wait-free implementations are particularly appealing since they ensure that each operation always terminates in a finite number of~steps.

Alas, known MPMC wait-free queues suffer from large memory overheads, intricate code complexity, and low scalability.
In particular, it was shown that wait-free MPMC queues require the use of a helper mechanism~\cite{help}.
On the other hand, as discussed above, there are important classes of applications for which MPSC queues are adequate.
Such applications could therefore benefit if a more efficient MPSC queue construction was found.
This motivates studying wait-free MPSC queues, which is the topic of this paper.

\paragraph*{Contributions} 

In this work we present \MYQUEUE{}, a fast memory efficient wait-free MPSC queue.
\MYQUEUE{} is unbounded in the the number of elements that can be enqueued without being dequeued (up to the memory limitations of the machine).
Yet the amount of memory \MYQUEUE{} consumes at any given time is proportional to the number of such items and \MYQUEUE{} minimizes the use of pointers, to reduce its memory~footprint.

To obtain these good properties, \MYQUEUE{} stores elements in a linked list of arrays, and only allocates a new array when the last array is being filled.
Also, as soon as all elements in a given array are dequeued, the array is released.
This way, a typical enqueue operation requires little more than a simple FAA and setting the corresponding entry to the enqueued value and changing its status from \texttt{empty} to \texttt{set}.
Hence, operations are very fast and the number of pointers is a multiple of the allocated arrays rather than the number of queued~elements.

To satisfy linearizability and wait-freedom, a dequeue operation in \MYQUEUE{} may return a value that is already past the head of the queue, if the enqueue operation working on the head is still on-going.
To ensure correctness, we devised a novel mechanism to handle such entries both during their immediate dequeue as well as during subsequent dequeues.

Another novel idea in \MYQUEUE{} is related to its buffer allocation policy.
In principle, when the last buffer is full, the naive approach is for each enqueuer at that point to allocate a new buffer and then try adding it to the queue with a CAS.
When multiple enqueuers try this concurrently, only one succeeds and the others need to free their allocated buffer.
However, this both creates contention on the end of the queue and wastes CPU time in allocating and freeing multiple buffers each time.
To alleviate these phenomena, in \MYQUEUE{} the enqueuer of the second entry in the last buffer already allocates the next buffer and tries to add it using CAS.
This way, almost always, when enqueuers reach the end of a buffer, the next buffer is already available for them without any contention.

We have implemented \MYQUEUE{} and evaluated its performance in comparison with three other leading lock-free and wait-free implementations, namely WFqueue~\cite{yang2016wait}, CCqueue~\cite{fatourou2012revisiting}, and MSqueue~\cite{MS96}.
We also examined the memory requirements for the data and code of all measured implementations using valgrind~\cite{nethercote2003valgrind}.
The results indicate that \MYQUEUE{} is up to $50$\% faster than WFqueue and roughly $10$ times times faster than CCqueue and MSqueue.
\MYQUEUE{} is also more scalable than the other queue structures we tested, enabling more than $20$ million operations per second even with 128 threads.
Finally, the memory footprint of \MYQUEUE{} is roughly $90$\% better than its competitors in the tested workloads, and provides similar benefits in terms of number of cache and heap accesses.
\MYQUEUE{} obtains better performance since the size of each queue node is much smaller and there are no auxiliary data structures.
For example, in WFqueue, which also employs a linked list of arrays approach, each node maintains two pointers, there is some per-thread meta-data, the basic slow-path structure (even when empty), etc.
Further, WFqueue employs a lazy reclamation policy, which according to its authors is significant for its performance.
Hence, arrays are kept around for some time even after they are no longer useful. 
In contrast, the per-node meta-data in \MYQUEUE{} is just a 2-bit flag, and arrays are being freed as soon as they become empty.
This translates to a more (hardware) cache friendly access pattern (as is evident in Tables~\ref{tab:2threads} and~\ref{tab:128threads}).
Also, in \MYQUEUE{} dequeue operations do not invoke any atomic (e.g., FAA \& CAS) operations at all.

\section{Related Work}
\label{sec:related}
Implementing concurrent queues is a widely studied topic~\cite{alistarh2015spraylist,AH90,DDGZ18,DBLP:conf/usenix/HedayatiSSM19,liu2007active,SLS06,DBLP:conf/cf/StratiGSGK19,vandierendonck2013deterministic}.
Below we focus on the most relevant works.

\noindent\textbf{Multi-Multi Queues:}
One of the most well known lock-free queue constructions was presented by Michael and Scott~\cite{MS96}, aka \emph{MSqueue}.
It is based on a singly-linked list of nodes that hold the enqueued values plus two references to the head and the tail of the list.
Their algorithm does not scale past a few threads due to contention on the queue's head and tail.

Kogan and Petrank introduced a wait-free variant of the MSqueue~\cite{kogan2011wait}.
Their queue extends the helping technique already employed by Michael and Scott to achieve wait freedom with similar performance characteristics.
They achieve wait-freedom by assigning each operation a dynamic age-based priority and making threads with younger operations help older operations to complete through the use of another data structure named a \emph{state}~array. 

Morrison and Afek proposed LCRQ~\cite{morrison2013fast}, a non-blocking queue based on a linked-list of circular ring segments, CRQ for short. 
LCRQ uses FAA to grab an index in the CRQ.
Enqueue and dequeue operations in~\cite{morrison2013fast} involve a double-width compare-and-swap (CAS2).

Yang and Mellor-Crummey proposed WFqueue, a wait free queue based on FAA~\cite{yang2016wait}.
WFqueue utilizes a linked-list of fixed size segments.
Their design employs the fast-path-slow-path methodology~\cite{kogan2012methodology} to transform a FAA based queue into a wait-free queue.
That is, an operation on the queue first tries the fast path implementation until it succeeds or the number of failures exceeds a threshold.
If necessary, it falls back to the slow-path, which guarantees completion within a finite number of attempts.
Each thread needs to register when the queue is started, so the number of threads cannot change after the initialization of the queue.
In contrast, in our queue a thread can join anytime during the run.

Fatourou and Kallimanis proposed CCqueue~\cite{fatourou2012revisiting}, a blocking queue that uses combining.
In CCqueue, a single thread scans a list of pending operations and applies them to the queue.
Threads add their operations to the list using SWAP. 

Tsigas and Zhang proposed a non-blocking queue that allows the head and tail to lag at most $m$ nodes behind the actual head and tail of the queue~\cite{tsigas2001simple}.
Additional cyclic array queues are described in~\cite{giacomoni2008fastforward,shafiei2009non}.
Recently, a lock-free queue that extends MSqueue~\cite{MS96} to support batching operations was presented in~\cite{DBLP:conf/spaa/MilmanKLLP18}.

\noindent\textbf{Limited Concurrency Queues:}
David proposed a sublinear time wait-free queue~\cite{david2004single} that supports multiple dequeuers and one enqueuer.
His queue is based on infinitely large arrays.
The author states that he can bound the space requirement, but only at the cost of increasing the time complexity to \emph{O(n)}, where \emph{n} is the number of dequeuers.

Jayanti and Petrovic proposed a wait-free queue implementation supporting multiple enqueuers and one concurrent dequeuer~\cite{jayanti2005logarithmic}.
Their queue is based on a binary tree whose leaves are linear linked lists. 
Each linked list represents a ``local'' queue for each thread that uses the queue. 
Their algorithm keeps one local queue at each process and maintains a timestamp for each element to decide the order between the elements in the different queues.

\section{Preliminaries}
\label{sec:prelim} 

We consider a standard shared memory setting with a set of threads accessing the shared memory using only the following atomic~operations:
\begin{itemize}
	\item Store - Atomically replaces the value stored in the target address with the given value.
	\item Load - Atomically loads and returns the current value of the target address.
	\item CAS - Atomically compares the value stored in the target address with the expected value. If those are equal, replaces the former with the desired value and the Boolean value \texttt{true} is returned. Otherwise, the shared memory is unchanged and \texttt{false} is returned.
	\item FAA -Atomically adds a given value to the value stored in the target address and returns the value in the target address held previously.
\end{itemize}

We assume that a program is composed of multiple such threads, which specifies the order in which each thread issues these operations and the objects on which they are invoked.
In an execution of the program, each operation is invoked, known as its \emph{invocation} event, takes some time to execute, until it terminates, known as its \emph{termination} event.
Each termination event is associated with one or more values \emph{being returned by that operation}.
An execution is called \emph{sequential} if each operation invocation is followed immediately by the same operation's termination (with no intermediate events between them).
Otherwise, the execution is said to be \emph{concurrent}.
When one of the operations in an execution $\sigma$ is being invoked on object $x$, we say that $x$ is being \emph{accessed} in $\sigma$.

Given an execution $\sigma$, we say that $\sigma$ induces a (partial) \emph{real-time ordering} among its operations:
Given two operations $o_1$ and $o_2$ in $\sigma$, we say that $o_1$ \emph{appears before} $o_2$ in $\sigma$ if the invocation of $o_2$ occurred after the termination of $o_1$ in  $\sigma$.
If neither $o_1$ nor $o_2$ appears before the other in $\sigma$, then they are considered \emph{concurrent operations} in $\sigma$.
Obviously, in a sequential execution the real-time ordering is a total order.

Also, we assume that each object has a \emph{sequential specification}, which defines the set of allowed sequential executions on this object.
A sequential execution $\sigma$ is called \emph{legal} w.r.t. a given object $x$ if the restriction of this execution to operations on $x$ (only) is included in the sequential specification of $x$.
$\sigma$ is said to be legal if it is legal w.r.t. any object being accessed in $\sigma$.
Further, we say that two executions $\sigma$ and $\sigma`$ are \emph{equivalent} if there is a one-to-one mapping between each operation in $\sigma$ to an operation in $\sigma'$ and the values each such pair of operations return in $\sigma$ and $\sigma'$ are the same.

\begin{definition}[Linearizability]
An execution $\sigma$ is \emph{linearizable}~\cite{herlihy1990linearizability} if it is equivalent to a legal sequential execution $\sigma'$ and the order of all operations in $\sigma'$ respects the real-time order of operations in $\sigma$.
\end{definition}

\section{MPSC Queue}
\label{sec:mpscqueue}

One of the problems of multi producer multi consumer queues is that any wait-free implementation of such queues requires utilizing helper data structures~\cite{help}, which are both memory wasteful and slow down the rate of operations.
By settling for single consumer support, we can avoid helper data structures.
Further, for additional space efficiency, we seek solutions that minimize the use of pointers, as each pointer consumes $64$ bits on most modern architectures.
The requirement to support unbounded queues is needed since there can be periods in which the rate of enqueue operations surpasses the rate of dequeues.
Without this property, during such a period some items might be dropped, or enqueuers might need to block until enough existing items are dequeued.
In the latter case the queue is not wait-free.

\subsection{Overview}

Our queue structure is composed of a linked list of buffers.
A buffer enables implementing a queue without pointers, as well as using fast \emph{FAA} instructions to manipulate the head and tail indices of the queue.
However, a single buffer limits the size of the queue, and is therefore useful only for bounded queue implementations.
In order to support unbounded queues, we extend the single buffer idea to a linked list of buffers.
This design is a tradeoff point between avoiding pointers as much as possible while supporting an unbounded queue.
It also decreases the use of \emph{CAS} to only adding a new buffer to the queue, which is performed~rarely.

Once a buffer of items has been completely read, the (sole) consumer deletes the buffer and removes it from the linked list.
Hence, deleting obsolete buffers requires no synchronization.
When inserting an element to the queue, if the last buffer in the linked list has room, we perform an atomic \emph{FAA} to a \emph{tail} index and insert the element.
Otherwise, the producer allocates a new buffer and tries to insert it to the linked list via an atomic \emph{CAS}, which keeps the overall size of the queue small while supporting unbounded sizes\footnote{In fact, to aviod contention the new buffer is usually already allocated earlier on; more accurate details appear below.}.
The structure of a \MYQUEUE{} queue is depicted in Figure~\ref{fig:queue}.
\begin{figure}[t]
	\center{
		\includegraphics[width =0.55\columnwidth]{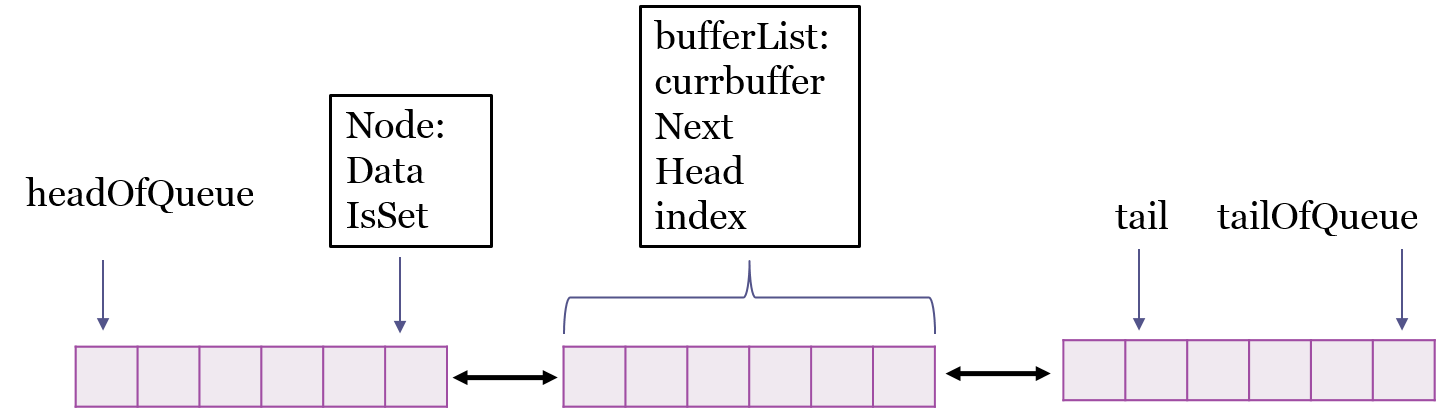}
	}
	\caption{A node in \MYQUEUE{} consists of two fields: the data itself and an \emph{IsSet} field that notifies when the data is ready to be read.
		The \emph{BufferList} consists of $5$ fields: \emph{Currbuffer} is an array of nodes, a \emph {Next} and \emph{Prev} pointers, \emph{Head} index pointing to the last place in that buffer that the consumer read, and \emph{PositionInQueue} that tracks the location of the \emph{BufferList} in the list.
		The \emph{Tail} index indicates the last place to be written to by the producers.
		\emph{HeadOfQueue} is the consumer pointer for the first buffer; once the head reaches the end of the buffer, the corresponding \emph{BufferList} is deleted.
		\emph{TailOfQueue} points to the last BufferList in the linked list; once the \emph{Tail} reaches the end of this buffer a new array BufferList is added.
	}
	\label{fig:queue}
\end{figure}

\ifdefined\COVER
Yet, the queue as described so far is not \emph{linearizable}.
For example, consider two concurrent enqueue operations $enq_1$ and $enq_2$ for items $i_1$ and $i_2$ respectively, where $enq_1$ terminates before $enq_2$.
Also, assume a dequeue operation that overlaps only with $enq_2$ (it starts after $enq_1$ terminates).
It is possible that item $i_2$ is inserted at an earlier index in the queue than the index of $i_1$ because this is the order by which the \emph{FAA} instructions of $enq_1$ and $enq_2$ executed.
Yet, the insertion of $i_1$ terminates quickly while the insertion of $i_2$ (into the earlier index) takes longer.
Now the dequeue operation sees that the \emph{head} of the queue is still empty, so it returns immediately with an empty reply.
But since $enq_1$ terminated before the dequeue started, this violates linearizability.
\else
Figure~\ref{fig:linearizable} exhibits why the naive proposal for the queue violates linearizability.
This scenario depicts two concurrent enqueue operations $enqueue_1$ and $enqueue_2$ for items $i_1$ and $i_2$ respectively, where $enqueue_1$ terminates before $enqueue_2$.
Also, assume a dequeue operation that overlaps only with $enqueue_2$ (it starts after $enqueue_1$ terminates).
It is possible that item $i_2$ is inserted at an earlier index in the queue than the index of $i_1$ because this is the order by which their respective \emph{FAA} instructions got executed.
Yet, the insertion of $i_1$ terminates quickly while the insertion of $i_2$ (into the earlier index) takes longer.
Now the dequeue operation sees that the \emph{head} of the queue is still empty, so it returns immediately with an empty reply.
But since $enqueue_1$ terminates before the dequeue starts, this violates linearizability.
\begin{figure}[H]
	\center{
		\includegraphics[width = 0.22\textwidth]{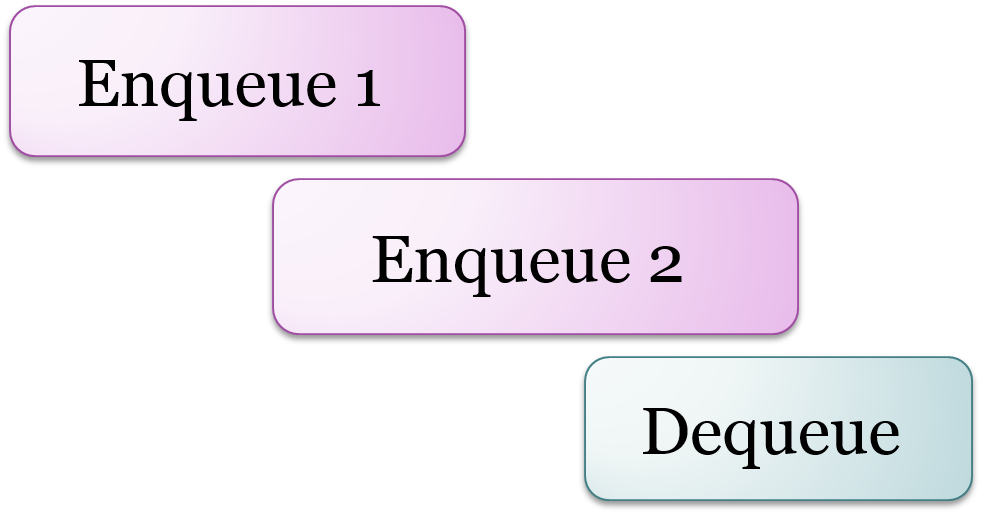}
	}
	\caption{A linearizability problem in the basic queue idea.}
	\label{fig:linearizable}
\end{figure}
\fi
Yet, the queue as described so far is not \emph{linearizable}, as exemplified in Figure~\ref{fig:linearizable}.
To make the queue linearizable, 
if the dequeuer finds that the \emph{isSet} flag is not \texttt{set} (nor \texttt{handled}, to be introduced shortly), as illustrated in Figure~\ref{fig:handeld1}, the dequeuer continues to scan the queue until either reaching the \emph{tail} or finding an item $x$ at position $i$ whose \emph{isSet} is \texttt{set} but the \emph{isSet} of its previous entry is not, as depicted in Figure~\ref{fig:handeld2}.
In the latter case, the dequeuer verifies that the \emph{isSet} of all items from the head until position $i$ are still neither \texttt{set} nor \texttt{handled} and then returns $x$ and sets the \emph{isSet} at position $i$ to \texttt{handled} (to avoid re-dequeuing it in the future).
This double reading of the \emph{isSet} flag is done to preserve linearizability in case an earlier enqueuer has terminated by the time the newer item is found.
In the latter case the dequeue should return the earlier item to ensure correct ordering.
This becomes clear in the correctness~proof.

\begin{figure}[t]
\begin{subfigure}{0.5\textwidth}
	\center {
		\includegraphics[width = 0.5\columnwidth]{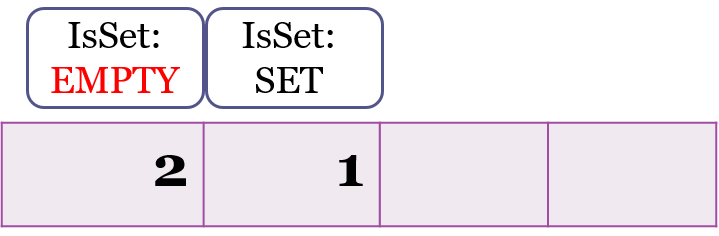}	
	}
	\caption{
		The dequeuer finds that the \emph{isSet} flag for the element pointed by the head is Empty.
	}
	\label{fig:handeld1}
\end{subfigure}
\begin{subfigure}{0.5\textwidth}
	\center {
		\includegraphics[width = 0.5\columnwidth]{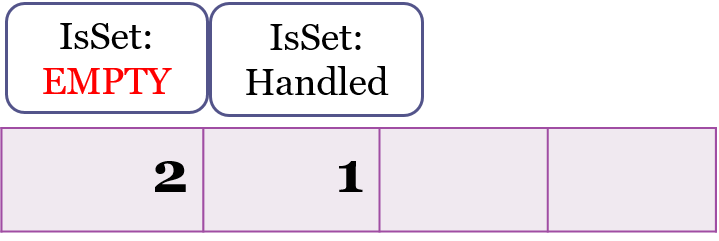}	
	}
	\caption{
		The dequeuer found an item whose \emph{isSet} is ready, the \emph{isSet} of the entry pointed by head is not set. It removes this item and marks its \emph{isSet} as \texttt{handled}.
	}
	\label{fig:handeld2}
\end{subfigure}
\caption{An example of the solution for the basic queue's linearizability problem.}
\label{fig:lin-problem-solution}
\end{figure}

 \begin{algorithm}[t]
	\caption{Internal \MYQUEUE{} classes and fields}
	\label{alg:vars}
	\scriptsize
	\begin{algorithmic}[1]
		\State class Node $\{$ \label{line:kuku}
		\State $\ \ T\ data;$
		\State 	$\ \ atomic<State> \  isSet;\}$
		\State  $// \ isSet $ has three values: empty, set and handled.
		\State class bufferList $\{$
		\State $\ \  Node* \ currBuffer;$
		\State 	$\ \ atomic<bufferList*>\  next;$
		\State $\ \ bufferList*\ prev;$
		\State $\ \ unsigned\ int\ head;$
		\State $\ \ unsigned\ int\ positionInQueue;\}$		
		\State $bufferList*\ headOfQueue;$
		\State $atomic<bufferList*>\ tailOfQueue;$
		\State $atomic<unsigned\ int>\ tail;$		
	\end{algorithmic}
\end{algorithm}


\subsection{\MYQUEUE{} Queue Algorithm}

\subsubsection{Internal \MYQUEUE{} classes and fields}
Algorithm~\ref{alg:vars} depicts the internal classes and variables.
The \emph{Node} class represents an item in the queue.
It holds the data itself and a flag to indicate whether the data is ready to be read (\emph{isSet}=\texttt{set}), the node is empty or is in inserting process (\emph{isSet}=\texttt{empty}), or it was already read by the dequeuer (\emph{isSet}=\texttt{handled}). 
Every node in the queue starts with \emph{isSet}=\texttt{empty}.
 
The \emph{bufferList} class represents each buffer in the queue:
\emph{currbuffer} is the buffer itself -- an array of nodes;
\emph{next} is a pointer to the next buffer -- it is atomic as several threads (the enqueuers) try to change it simultaneously;
\emph{prev} is a pointer to the previous buffer -- it is not concurrently modified and therefore not atomic;
\emph{head} is an index to the first location in the buffer to be read by the dequeuer -- it is changed only by the single dequeuer thread;
\emph{positionInQueue} is the position of the buffer in the queue -- it is used to calculate the amount of items in the queue up until this buffer, it never changes once set and is initialized to~1.

The rest are fields of the queue:
\emph{headOfQueue} points to the head buffer of the queue. 
The dequeuer reads from the buffer pointed by \emph{headOfQueue} at index \emph{head} in that buffer.
It is only changed by the single threaded dequeuer.
\emph{tailOfQueue} points to the last queue's buffer.
It is atomic as it is modified by several threads.
\emph{tail}, initialized to 0, is the index of the last queued item. 
All threads do FAA on tail to receive a location of a node to insert~into.

\begin{algorithm}[t]
	\caption{Enqueue operation}
	\label{alg:enqueue_high}
	\scriptsize
	\begin{algorithmic}[1]
	\Function {enqueue}{$data$}
	\State{$location=FAA(tail);$} \label{alg:HFAA}
		
	\While{the $location$ is in an unallocated buffer }\label{alg:Hnext}
	\State{allocate a new buffer and try to adding it to the queue with a CAS on $tailOfQueue$}	
	\If{unsuccessful}
	\State{delete the allocated buffer and move to the new buffer}\label{alg:Hdelete}

	\EndIf	
	\EndWhile
	
		\State{bufferList* tempTail = tailOfQueue}

		\While{the $location$ is not in the buffer pointed by $tempTail$  } \label{alg:Hback}
			\State{tempTail = tempTail$\rightarrow$ prev}
			\EndWhile

		\State{//location is in this buffer }
			\State{adjust location to its corresponding index in $tempTail$ } \label{line:Hadjust}
			\State{$tempTail[location ].data=data$} \label{line:enq-location}
			\State{$tempTail[location ].isSet=\mathtt{set}$} \label{line:enq-set}
			\If{$location$ is the second entry of the last buffer}\label{alg:Hindexone}
				\State{allocate a new buffer and try adding it with a CAS on $tailOfQueue$; if unsuccessful, delete this buffer}	
			\EndIf

	\EndFunction
		
	\end{algorithmic}
	
\end{algorithm}

\subsubsection{The Enqueue Operation}
A high-level pseudo-code for enqueue operations appears in Algorithm~\ref{alg:enqueue_high}; a more detailed listing appears in Appendix~\ref{sec:detailedcode}.
The enqueue operation can be called by multiple threads acting as producers.
The method starts by performing FAA on \emph{Tail} to acquire an index in the queue (Line~\ref{alg:HFAA}).
When there is a burst of enqueues the index that is fetched can be in a buffer that has not yet been added to the queue, or a buffer that is prior to the last one.
Hence, the enqueuer needs to find in which buffer it should add its new item.

If the index is beyond the last allocated buffer, the enqueuer allocates a new buffer and tries to add it to the queue using a \emph{CAS} operation (line~\ref{alg:Hnext}).
If several threads try simultaneously, one \emph{CAS} will succeed and the other enqueuers' \emph{CAS} will fail.
Every such enqueuer whose \emph{CAS} fails must delete its buffer (line~\ref{alg:Hdelete}).
If the \emph{CAS} succeeds, the enqueuer moves \emph{tailOfQueue} to the new buffer. 
If there is already a next buffer allocated then the enqueuer only moves the \emph{tailOfQueue} pointer through a CAS operation.

If the index is earlier in the queue compared to the tail index (line~\ref{alg:Hback}), the enqueuer retracts to the previous buffer.
When the thread reaches the correct buffer, it stores the data in the buffer index it fetched at the beginning  (line~\ref{line:enq-location}), marks the location as \texttt{Set} (line~\ref{line:enq-set}) and finishes.
Yet, just before returning, if the thread is in the last buffer of the queue and it obtained the second index in that buffer, the enqueuer tries to add a new buffer to the end of the queue (line~\ref{alg:Hindexone}).
This is an optimization step to prevent wasteful contention prone allocations of several buffers and the deletion of most as will be explain next.

Notice that if we only had the above mechanism, then each time a buffer ends there could be a lot of contention on adding the new buffer, and many threads might allocate and then delete their unneeded buffer.
This is why the enqueuer that obtains the second entry in each buffer already allocates the next one.
With this optimization, usually only a single enqueuer tries to allocate such a new buffer and by the time enqueuers reach the end of the current buffer a new one is already available to them.
On the other hand, we still need the ability to add a new buffer if one is not found (line~\ref{alg:Hnext}) to preserve wait-freedom.



\begin{algorithm}[t]
	\caption{Dequeue operation}
	\label{alg:dequeue_high}
	\scriptsize
	\begin{algorithmic}[1]
		\Function {dequeue}{}		
		\State {mark the element pointed by head as \emph{n}} 
		\State{//Skip to the first non-\texttt{handled} element (due to the code below, it might not be pointed by head!)}
		\While{\emph{n}.isSet == \texttt{handled}} \label{alg:HHandled}
			\State {advance n and head to the next element} 
				
		\If{the entire buffer has been read}
			\State {move to the next buffer if exists and delete the previous buffer}\label{alg:Hmove}
		\EndIf	
			
		\EndWhile
		\If{queue is empty} \label{alg:Hempty}
			\State {return false} 
		\EndIf
		
		\State {// If the queue is not empty, but its first element is, there might be a Set element further on -- find it}
		\If{\emph{n}.isSet == \texttt{empty}}\label{alg:Hwhile}
			\For {(\emph{tempN} = \emph{n}; \emph{tempN}.isSet != \texttt{set} and not end of queue; advance \emph{tempN} to next element)} \label{alg:HforN}
				\If{the entire buffer is marked with \texttt{handled}}
					\State {``fold'' the queue by deleting this buffer and move to the next buffer if exists}\label{alg:Hfold}
					\State \Comment{Notice comment about a delicate garbage collection issue in the description text}
				\EndIf	
			\EndFor
			\If{reached end of queue} \label{alg:HendQ}
				\State {return false}
			\EndIf
		\EndIf
			\State {// Due to concurrency, some element between \emph{n} and \emph{tempN} might have been \texttt{set} -- find it}
			\For{(\emph{e} = n; \emph{e} = \emph{tempN} or \emph{e} is before \emph{tempN}; advance \emph{e} to next element)}\label{alg:forEach}
			 \If {\emph{e} is not \emph{n} and \emph{e}.isSet == \texttt{set}}
			 	\State { \emph{tempN} =\emph{e} and restart the for loop again from \emph{n} } \label{alg:Hscan}
		 	\EndIf
		 	\EndFor
		 		\State {// we scanned the path from head to \emph{tempN} and did not find a prior \texttt{set} element - remove \emph{tempN}  }
				\State{\emph{tempN}.isSet= \texttt{handled}} \label{alg:HremoveN}
				\If{the entire buffer has been read}
					\State {move to the next buffer if exists and delete the previous buffer}\label{alg:Hmove2}
				\ElsIf {\emph{tempN} = \emph{n}} \label{alg:HadvanceHead}
					\State{advance head}
				\EndIf
				\State{return \emph{tempN}.data}
%
	
%
%
%
%
%
%
		\EndFunction
	\end{algorithmic}
\end{algorithm}

\begin{figure}[t]
	\center{
		\includegraphics[width =1\textwidth]{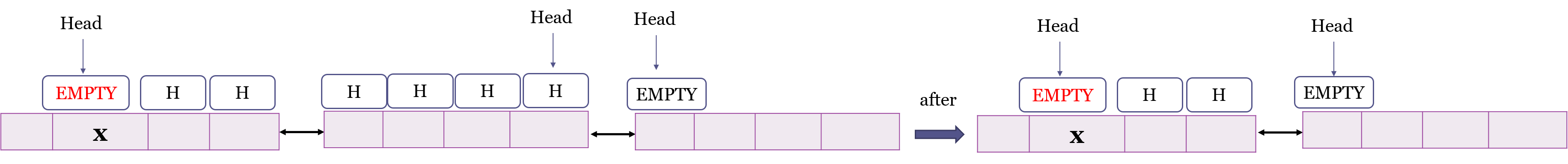}	
	}
	\caption{``folding'' the queue - A thread fetches an index in the queue and stalls before completing the enqueue operation. Here we only keep the buffer that contains this specific index and delete the rest.
		H stands for \emph{isSet} = \texttt{handled} and EMPTY stands for \emph{isSet} = \texttt{empty}.
	}
	\label{fig:remove2}
\end{figure}

\subsubsection{The Dequeue Operation}
A high-level pseudo-code for dequeue operations appears in Algorithm~\ref{alg:dequeue_high}; a more detailed listing appears in Appendix~\ref{sec:detailedcode}.
A dequeue operation is called by a single thread, the consumer.
A dequeue starts by advancing the head index to the first element not marked with \textit{isSet} = \texttt{handled} as such items are already dequeued (line~\ref{alg:HHandled}).
If the consumer reads an entire buffer during this stage, it deletes it (line~\ref{alg:Hmove}).
At the end of this scan, the dequeuer checks if the queue is empty and if so returns false (line~\ref{alg:Hempty}).

Next, if the first item is in the middle of an enqueue process (\textit{isSet}=\texttt{empty}), the consumer scans the queue (line~\ref{alg:HforN}) to avoid the linearizability pitfall mentioned above.
If there is an element in the queue that is already \texttt{set} while the element pointed by the head is still \texttt{empty}, then the consumer needs to dequeue the latter item (denoted \textit{tempN} in line~\ref{alg:HforN}).

During the scan, if the consumer reads an entire buffer whose cells are all marked \texttt{handled}, the consumer deletes this buffer (line~\ref{alg:Hfold}).
We can think of this operation as ``folding'' the queue by removing buffers in the middle of the queue that have already been read by the consumer. 
Hence, if a thread fetches an index in the queue and stalls before completing the enqueue operation, we only keep the buffer that contains this specific index and may delete all the rest, as illustrated in Figure~\ref{fig:remove2}.

Further, before dequeuing, the consumer scans the path from head to \textit{tempN} to look for any item that might have changed it status to \texttt{set} (line~\ref{alg:Hscan}).
If such an item is found, it becomes the new \textit{tempN} and the scan is restarted.

Finally, a dequeued item is marked \texttt{handled} (line~\ref{alg:HremoveN}), as also depicted in Figure~\ref{fig:handeld2}.
Also, if the dequeued item was the last non-\texttt{handled} in its buffer, the consumer deletes the buffer and moves the head to the next one (line~\ref{alg:HadvanceHead}).

There is another delicate technical detail related to deleting buffers in non-garbage collecting environments such as the run-time of C++.
For clarity of presentation and since it is only relevant in non-garbage collecting environments, the following issue is not addressed in Figure~\ref{alg:dequeue_high}, but is rather deferred to the detailed description in Figure~\ref{alg:dequeue} at the Appendix.
Specifically, when preforming the fold, only the array is deleted, which consumes the most memory, while the much smaller meta-data structure of the array is transferred to a dedicated garbage collection list.
The reason for not immediately deleting the the entire array's structure is to let enqueuers, who kept a pointer for tailOfQueue point to a valid bufferList at all times, with the correct prev and next pointers, until they are done with it.
The exact details of this garbage collection mechanism are discussed in Appendix~\ref{sec:detailedcode}.

\ifdefined\COVER
\else
\subsubsection{Memory Buffer Pool Optimization}
\label{sec:buffer-pool}
Instead of always allocating and releasing buffers from the operating system, we can maintain a buffer pool.
This way, when trying to allocate a buffer, we first check if there is already a buffer available in the buffer pool.
If so, we simply claim it without invoking an OS system call.
Similarly, when releasing a buffer, rather than freeing it with an OS system call, we can insert it into the buffer pool.
It is possible to have a single shared buffer pool for all threads, or let each thread maintain its own buffer pool.
This optimization can potentially reduce execution time at the expense of a somewhat larger memory heap area.
The code used for \MYQUEUE's performance measurements does \emph{not} employ this optimization.
In \MYQUEUE, allocation and freeing of buffers are in any case a relatively rare event.
\fi

\section{Correctness}
\label{sec:proof}

\subsection{Linearizability}

To prove linearizability, we need to show that for each execution $\sigma$ that may be generated by \MYQUEUE, we can find an equivalent legal sequential execution that obeys the real-time order of operations in $\sigma$.
We show this by constructing such an execution.
Further, for simplicity of presentation, we assume here that each value can be enqueued only once.
Consequently, it is easy to verify from the code that each value can also be returned by a dequeue operation at most once.
In summary we have:

\begin{observation}
	\label{obs:once}
	Each enquequed value can be returned by a dequeue operation at most once.
\end{observation}

\begin{theorem}
	The \MYQUEUE{} queue implementation is linearizable.
\end{theorem}

\begin{proof}
Let $\sigma$ be an arbitrary execution generated by \MYQUEUE.
We now build an equivalent legal sequential execution $\sigma'$.
We start with an empty sequence of operations $\sigma'$ and gradually add to it all operations in $\sigma$ until it is legal and equivalent to $\sigma$.
Each operation inserted into $\sigma'$ is collapsed such that its invocation and termination appear next to each other with no events of other operations in between them, thereby constructing $\sigma'$ to be sequential.

First, all dequeue operations of $\sigma$ are placed in $\sigma'$ in the order they appear in $\sigma$.
Since \MYQUEUE{} only supports a single dequeuer, then in $\sigma$ there is already a total order on all dequeue operations, which is preserved in $\sigma'$.

Next, by the code, a dequeue operation $deq$ that does not return \texttt{empty}, can only return a value that was inserted by an enqueue operation $enq$ that is concurrent or prior to $deq$ (only enqueue operations can change an entry to \texttt{set}).
By Observation~\ref{obs:once}, for each such $deq$ operation there is exactly one such $enq$ operation.
Hence, any ordering in which $enq$ appears before $deq$ would preserve the real-time ordering between these two operations.

Denote the set of all enqueue operations in $\sigma$ by $ENQ$ and let $\widehat{ENQ}$ be the subset of $ENQ$ consisting of all enqueue operations $enq$ such that the value enqueued by $enq$ is returned by some operation $deq$ in $\sigma$.
Next, we order all enqueue operations in $\widehat{ENQ}$ in the order of the dequeue operations that returned their respective value.

\begin{claim}
	\label{claim:one}
The real time order in $\sigma$ is preserved among all operations in $\widehat{ENQ}$.
\end{claim}
\begin{proof}[Proof of Claim~\ref{claim:one}]
Suppose Claim~\ref{claim:one} does not hold.
Then there must be two enqueue operations $enq_1$ and $enq_2$ and corresponding dequeue operations $deq_1$ and $deq_2$ such that the termination of $enq_1$ is before the invocation of $enq_2$ in $\sigma$, but $deq_2$ occurred before $deq_1$.
Denote the entry accessed by $enq_1$ by $in_1$ and the entry accessed by $enq_2$ by $in_2$.
Since the invocation of $enq_2$ is after the termination of $enq_1$, then during the invocation of $enq_2$ the status of $in_1$ was already \texttt{set} (from line~\ref{line:enq-set} in Algorithm~\ref{alg:enqueue_high}).
Moreover, $in_2$ is further in the queue than $in_1$ (from line~\ref{alg:HFAA} in Algorithm~\ref{alg:enqueue_high}).
Since $deq_2$ returned the value enqueued by $enq_2$, by the time $deq_2$ accessed $in_2$ the status of $in_1$ was already \texttt{set}.
As we assumed $deq_1$ is after $deq_2$, no other dequeue operation has dequeueud the value in $in_1$.
Hence, while accessing $in_2$ and before terminating, $deq_2$ would have checked $in_1$ and would have found that it is now \texttt{set} (the loops in lines~\ref{alg:HHandled} and~\ref{alg:forEach} of Algorithm~\ref{alg:dequeue_high} -- this is the reason for line~\ref{alg:forEach}) and would have returned that value instead of the value at $in_2$.
A contradiction.
\end{proof}

To place the enqueue operations of $\widehat{ENQ}$ inside $\sigma'$, we scan them in their order in $\widehat{ENQ}$ (as defined above) from earliest to latest.
For each such operation $enq$ that has not been placed yet in $\sigma$:
($i$) let $deq'$ be the latest dequeue operation in $\sigma$ that is concurrent or prior to $enq$ in $\sigma$ and neither $deq'$ nor any prior dequeue operation return the value enqueued by $enq$, and
($ii$) let $deq''$ be the following dequeue operation in $\sigma$; we add $enq$ in the last place just before $deq''$ in $\sigma'$.
This repeats until we are done placing all operations from $\widehat{ENQ}$ into~$\sigma'$.

\begin{claim}
	\label{claim:two}
	The real-time ordering in $\sigma$ between dequeue operations and enqueue operations in $\widehat{ENQ}$ is preserved in $\sigma'$ as built thus far.
\end{claim}
\begin{proof}[Proof of Claim~\ref{claim:two}]
Since we placed each enqueue operation $enq$ before any dequeue operation whose invocation is after the termination of $enq$, we preserve real-time order between any pair of enqueue and dequeue operations.
Similarly, by construction the relative order of dequeue operations is not modified by the insertion of enqueue operations into $\sigma'$. Thus, real-time ordering is preserved.
\end{proof}

Hence, any potential violation of real time ordering can only occur by placing enqueue operations in a different order (with respect to themselves) than they originally appeared in $\widehat{ENQ}$.
For this to happen, it means that there are two enqueue operation $enq_1$ and $enq_2$ such that $enq_1$ appears before $enq_2$ in $\widehat{ENQ}$ but ended up in the reverse order in $\sigma'$.
Denote $deq'_1$ the latest dequeue operation in $\sigma$ that is prior or concurrent to $enq_1$ and neither $deq'_1$ nor any prior dequeue operation return the value enqueued by $enq_1$ and similarly denote $deq'_2$ for $enq_2$.
Hence, if $enq_2$ was inserted at an earlier location than $enq_1$, then $deq'_2$ is also before $deq'_1$.
But since the order of $enq_1$ and $enq_2$ in $\widehat{ENQ}$ preserves their real time order, the above can only happen if the value enqueued by $enq_2$ was returned by an earlier dequeue than the one returning the value enqueued by $enq_1$.
Yet, this violates the definition of the ordering used to create $\widehat{ENQ}$.

\begin{claim}
	\label{claim:three}
	The constructed execution $\sigma'$ preserves legality.
\end{claim}
\begin{proof}[Proof of Claim~\ref{claim:three}]
By construction, enqueue operations are inserted in the order their values have been dequeued, and each enqueue is inserted before the dequeue that returned its value.
Hence, the only thing left to show is legality w.r.t. dequeue operations that returned \texttt{empty}.

To that end, given a dequeue operation $deq_i$ that returns \texttt{empty}, denote $\#\widehat{deq}_{\sigma',i}$ the number of dequeue operations that did not return \texttt{empty} since the last previous dequeue operation that did return \texttt{empty}, or the beginning of $\sigma'$ of none exists.
Similarly, denote $\#enq_{\sigma',i}$ the number of enqueue operations during the same interval of~$\sigma'$.

\begin{subclaim}
	\label{claim:four}
	For each $deq_i$ that returns empty, $\#\widehat{deq}_{\sigma',i} > \#enq_{\sigma',i}$.
\end{subclaim}
\begin{proof}[Proof of Sub-Claim~\ref{claim:four}]
Recall that the ordering in $\sigma'$ preserves the real time order w.r.t. $\sigma$ and there is a single dequeuer.
Assume by way of contradiction that the claim does not hold, and let $deq_i$ be the first dequeue operation that returned \texttt{empty} while $\#\widehat{deq}_{\sigma',i} \leq \#enq_{\sigma',i}$.
Hence, there is at least one enqueue operation $enq_j$ in the corresponding interval of $\sigma'$ whose value is not dequeued in this interval.
In this case, $deq_i$ cannot be concurrent to $enq_j$ in $\sigma$ since otherwise by construction $enq_j$ would have been placed after it (as it does not return its value).
Hence, $deq_i$ is after $enq_j$ in $\sigma$.
Yet, since each dequeue removes at most one item from the queue, when $deq_i$ starts, the tail of the queue is behind the head and there is at least one item whose state is \texttt{set} between them.
Thus, by lines~\ref{alg:HHandled} and~\ref{alg:forEach} of Algorithm~\ref{alg:dequeue_high}, $deq_i$ would have returned one of these items rather then return \texttt{empty}.
A contradiction.
\end{proof}

With Sub-Claim~\ref{claim:four} we conclude the proof that $\sigma'$ as constructed so far is legal.
\end{proof}

The last thing we need to do is to insert enqueue operations whose value was not dequeued, i.e., all operations in $ENQ \setminus \widehat{ENQ}$.
Denote by $enq'$ the last operation in $\widehat{ENQ}$.
\begin{claim}
	\label{claim:five}
	Any operation $enq'' \in ENQ \setminus \widehat{ENQ}$ is either concurrent to or after $enq'$ in $\sigma$.
\end{claim}
\begin{proof}[Proof of Claim~\ref{claim:five}]
Assume, by way of contradiction, that there is an operation $enq'' \in ENQ \setminus \widehat{ENQ}$ that is before $enq'$ in $\sigma$.
Hence, by the time $enq'$ starts, the corresponding entry of $enq''$ is already \texttt{set} (line~\ref{line:enq-set} of Algorithm~\ref{alg:enqueue_high}) and $enq'$ obtains a later entry (line~\ref{alg:HFAA}).
Yet, since dequeue operations scan the queue from head to tail until finding a \texttt{set} entry (lines~\ref{alg:HHandled},~\ref{alg:HforN}, and~\ref{alg:forEach} in Algorithm~\ref{alg:dequeue_high}), the value enqueued by $enq''$ would have been dequeued before the value of $enq'$.
A contradiction.
\end{proof}

Hence, following Claim~\ref{claim:five}, we insert to $\sigma'$ all operations in $ENQ \setminus \widehat{ENQ}$ after all operations of $\widehat{ENQ}$.
In case of an enqueue operation $enq$ that is either concurrent with or later than a dequeue $deq$ in $\sigma$, then $enq$ is inserted to $\sigma'$ after $deq$.
Among concurrent enqueue operations, we break symmetry arbitrarily.
Hence, real-time order is preserved in~$\sigma'$.

The only thing to worry about is the legality of dequeue operations that returned \texttt{empty}.
For this, we can apply the same arguments as in Claim~\ref{claim:three} and Sub-claim~\ref{claim:four}.
In summary, $\sigma'$ is an equivalent legal sequential execution to $\sigma$ that preserves $\sigma$'s real-time order. 
\end{proof}

\begin{table*}[]
\resizebox{\textwidth}{!}{%
	\begin{tabular}{|l||l||l|l||l|l||l|l||l|l|}
		\hline
		\multirow{10}{*}{} &
		\multicolumn{1}{c||}{Jiffy} &
		\multicolumn{2}{c||}{WF} &
		\multicolumn{2}{c||}{LCRQ}&
		\multicolumn{2}{c||}{CC}&
		\multicolumn{2}{c|}{MS}\\
		& Absolute  & Absolute &Relative & Absolute &Relative & Absolute & Relative & Absolute & Relative \\
		\hline
		Total Heap Usage & 38.70 MB &611.63 MB & x15.80 & 610.95 MB&x15.78 & 305.18 MB & x7.88 &1.192 GB &x31.54 \\
		\hline
		Number of Allocs & 3,095 & 9,793 & x3.16 &1,230 &x0.40 &5,000,015 &x1,615 & 5,000,010 & x1,615 \\
		\hline
		Peak Heap Size & 44.81 MB  & 200.8 MB & x4.48 &612.6 MB &x13.67 &  420.0 MB &x9.37 &  1.229 GB &x28.08 \\
		\hline
		\# of Instructions Executed & 550,416,453 & 5,612,941,764 &x10.20 & 1,630,746,827 &x2.96 & 3,500,543,753 &x6.36 & 1,821,777,428 & x3.31  \\
		\hline
		I1  Misses & 2,162 &1,714 & x0.79 &1,601 &x0.74 & 1,577 & x0.73& 1,636 & x0.76 \\
		\hline
		L3i Misses &2,084 & 1,707 &x0.82 &1,591&x0.76 &1,572 & x0.75 &1,630 &x0.78 \\
		\hline
		Data Cache Tries (R+W)& 281,852,749 & 2,075,332,377 & x7.36 &  650,257,906 &x2.3 & 1,238,761,631 &x4.40 & 646,304,575 & x2.29 \\
		\hline
		D1  Misses & 1,320,401 & 25,586,262 &x19.37 &20,037,956 &x15.17 &15,000,507 &x11.36 & 11,605,064 & x8.79 \\
		\hline
		L3d Misses & 652,194 & 10,148,090 & x15.56 &5,028,204&x7.7 & 14,971,182 & x22.96 & 10,973,055 & x16.82 \\
		\hline
	\end{tabular}
}
	\normalsize
	\caption{
		Valgrind memory usage statistics run with one enqueuer and one dequeuer. I1/D1 is the L1 instruction/data cache respectively while L3i/L3d is the L3 instruction/data cache respectively.
	}
	\label{tab:2threads}

\end{table*}

\subsection{Wait-Freedom}
We show that each invocation of enqueue and dequeue returns in a finite number of steps.

\begin{lemma} \label{lemma-enqueue}
Each enqueue operation in Algorithm~\ref{alg:enqueue_high} completes in a finite number of steps.
\end{lemma}

\ifdefined\COVER
For lack of space, the technical proof of Lemma~\ref{lemma-enqueue} is deferred to the appendix.
\else
\begin{proof}
An enqueue operation, as listed in Algorithm~\ref{alg:enqueue_high}, consists of two while loops (line~\ref{alg:Hnext} and line~\ref{alg:Hback}), each involving a finite number of operations, a single FAA operation at the beginning (line~\ref{alg:HFAA}), and then a short finite sequence of operations from line~\ref{line:enq-location} onward.
Hence, we only need to show that the two while loops terminate in a finite number of iterations.

The goal of the first while loop is to ensure that the enqueuer only accesses an allocated buffer.
That is, in each iteration, if the \emph{location} index it obtained in line~\ref{alg:HFAA} is beyond the last allocated buffer, the enqueuer allocates a new buffer and tries to add it with a CAS to the end of the queue (line~\ref{alg:Hnext}).
Even if there are concurrent such attempts by multiple enqueuers, in each iteration at least one of them succeeds, so this loop terminates in a finite number of steps at each such enqueuer.

The next step for the enqueuer is to obtain the buffer corresponding to \emph{location}.
As mentioned before, it is possible that by this time the queue has grown due to concurrent faster enqueue operations.
Hence, the enqueuer starts scanning from \emph{tailOfQueue} backwards until reaching the correct buffer (line~\ref{alg:Hback}).
Since new buffers are only added at the end of the queue, this while loop also terminates in a finite number of steps regardless of concurrency.
\end{proof}

\fi

\begin{lemma}\label{lemma-dequeue}
Each dequeue operation in Algorithm~\ref{alg:dequeue_high} completes in a finite number of steps.	
\end{lemma}

\ifdefined\COVER
For lack of space, the technical proof of Lemma~\ref{lemma-dequeue} is deferred to the appendix.
\else
\begin{proof}
We prove the lemma by analyzing Algorithm~\ref{alg:dequeue_high}.
Consider the while loop at line~\ref{alg:HHandled}.
Here, we advance head to point to the first non-\texttt{handled} element, in case it is not already pointing to one.
As indicated before, the latter could occur if a previous dequeue $deq_1$ removed an element not pointed by head.
As we perform this scan only once and the queue is finite, the loop terminates within a finite number of iterations, each consisting of at most a constant number of operations.

If at this point the queue is identified as \texttt{empty}, which is detectable by comparing head and tail pointers and indices, then the operation returns immediately (line~\ref{alg:Hempty}).
The following for loop is at line~\ref{alg:HforN}.
Again, since the queue is finite, the next \texttt{set} element is within a finite number of entries away.
Hence, we terminate the for loop after a finite number of iterations.

The last for loop iteration at line~\ref{alg:forEach} scans the queue from its head to the first \texttt{set} element, which as mentioned before, is a finite number of entries away.
Yet, the for loop could be restarted in line~\ref{alg:Hscan}, so we need to show that such a restart can occur only a finite number of times.
This is true because each time we restart, we shorten the ``distance'' that the for loop at line~\ref{alg:forEach} needs to cover, and as just mentioned, this ``distance'' is finite to begin with.

The rest of the code from line~\ref{alg:HremoveN} onward is a short list of simple instructions.
Hence, each dequeue operation terminates in a finite number of steps.
\end{proof}

\fi

\begin{theorem}
The \MYQUEUE{} queue implementation is wait-free.
\end{theorem}

\begin{proof}
Lemmas~\ref{lemma-enqueue} and~\ref{lemma-dequeue} show that both enqueue and dequeue operations are wait-free.
Therefore, the queue implementation is wait-free.
\end{proof}

\section{Performance Evaluation}
\label{sec:eval}
We compare \MYQUEUE{} to several representative queue implementations in the literature:
Yang and Mellor-Crummey's queue~\cite{yang2016wait} is the most recent wait free FAA-based queue, denoted WFqueue;
Morrison and Afek LCRQ \cite{morrison2013fast} as a representative of nonblocking FAA-based queues;
Fatourou and Kallimanis CCqueue~\cite{fatourou2012revisiting} is a blocking queue based on the combining principle.
We also test Michael and Scott's classic lock-free MSqueue~\cite{MS96}.
We  include a microbenchmark that only preforms FAA on a shared variable.
This serves as a practical upper bound for the throughput of all FAA based queue~implementations.
Notice that \MYQUEUE{} performs FAA only during enqueue operations, but not during dequeues.
\ifdefined\COVER
Also, measurements of Jiffy do \emph{not} include a memory buffer pool optimization mentioned in Appendix~\ref{sec:additional-results}.
\else
Also, measurements of Jiffy do \emph{not} include a memory buffer pool optimization mentioned in Section~\ref{sec:buffer-pool}.
\fi

\noindent\textbf{Implementation:}
We implemented our queue algorithm in C++~\cite{jiffy-code}.
We compiled \MYQUEUE{} with g++ version 7.4.0 with -Os optimization level.
We use the C implementation provided by~\cite{yang2016wait} for the rest of the queues mentioned here.
They are compiled with GCC 4.9.2 with -Os optimization level.
The buffer size of our queue is 1620 entries. 
The segment size of Yang and Mellor-Crummey queue is $2^{10}$
and in LCRQ it is $2^{12}$, the optimal sizes according to their respective authors.

\noindent\textbf{Platforms:}
We measured performance on the following servers:
\begin{itemize}
	\item AMD PowerEdge R7425 server with two AMD EPYC 7551 Processors. 
	Each processor has $32$ 2.00GHz/2.55GHz cores, each of which multiplexes $2$ hardware threads, so in total this system supports $128$ hardware threads.
	The hardware caches include 32K L1 cache, 512K L2 cache and 8192K L3 cache, and there are 8 NUMA nodes, 4 per processor.
	
	\item Intel Xeon E5-2667 v4 Processor including $8$ 3.20GHz cores with 2 hardware threads, so this system supports $16$ hardware threads.
	The hardware caches include 32K L1 cache, 256K L2 cache and 25600K L3 cache.
	
\end{itemize}

\noindent\textbf{Methodology:}
We used two benchmarks: one that only inserts elements to the queue (enqueue only benchmark) whereas the second had one thread that only dequeued while the others only~enqueued.

In each experiment, x threads concurrently perform operations for a fixed amount of seconds, as specified shortly.
We tightly synchronize the start and end time of threads by having them spin-wait on a ``start'' flag.
Once all threads are created, we turn this flag on and start the time measurement.
To synchronize the end of the tests, each thread checks on every operation an ``end'' flag (on a while loop). 
When the time we measure ends we turn this flag on.
Each thread then counts the amount of finished operations it preformed and all are combined with FAA after the ``end'' flag is turned on.

In order to understand the sensitivity of our results to the run lengths, we measure both $1$ and $10$ seconds runs.
That is, the fixed amount of time between turning on the ``start'' and ``end'' flags is set to $1$ and $10$ seconds, respectively.
The graphs depict the throughput in each case, i.e., the number of  operations applied to the shared queue per second by all threads, measured in million operations per second (MOPS).
Each test was repeated 11 times and the average throughput is reported.
All experiments employ an initially empty queue.

With AMD PowerEdge for all queues we pinned each software thread to a different hardware thread based on the architecture of our server.
We started mapping to the first NUMA node until it became full, with two threads running on the same core.
We continue adding threads to the closest NUMA node until we fill all the processor hardware threads.
Then we move to fill the next processor in the same way.

For Intel Xeon server we test up to 32 threads without pinning whereas the server only has 16 hardware threads.
This results in multiple threads per hardware thread.

\begin{figure}[t]
	\begin{subfigure}{0.5\textwidth}
	\center{
		\includegraphics[width=0.93\columnwidth]{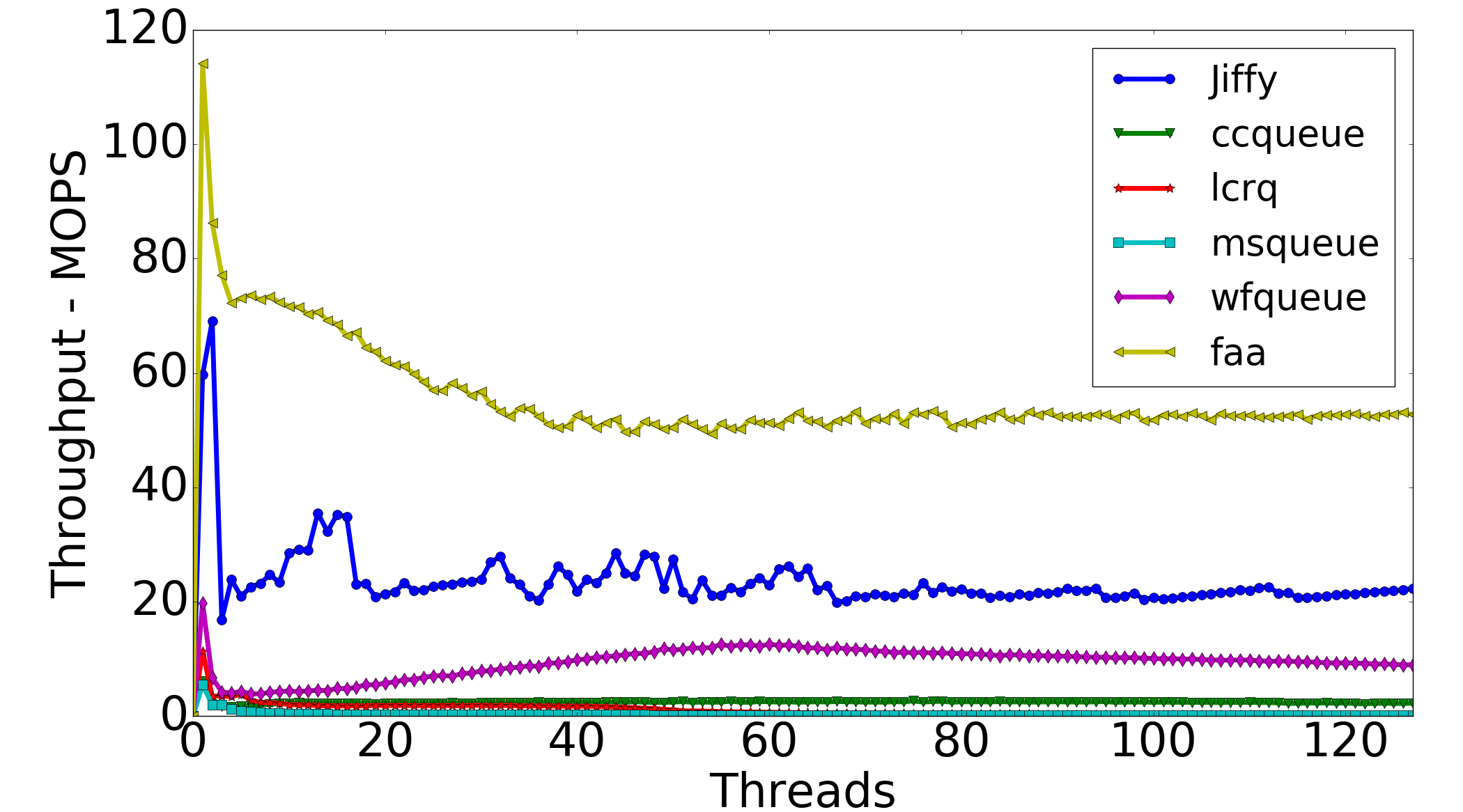}
	}
	\caption{
		AMD PowerEdge.
	}
	\label{fig:limonEnqueue}
\end{subfigure}
\begin{subfigure}{0.5\textwidth}
	\center{
		\includegraphics[width=0.93\columnwidth]{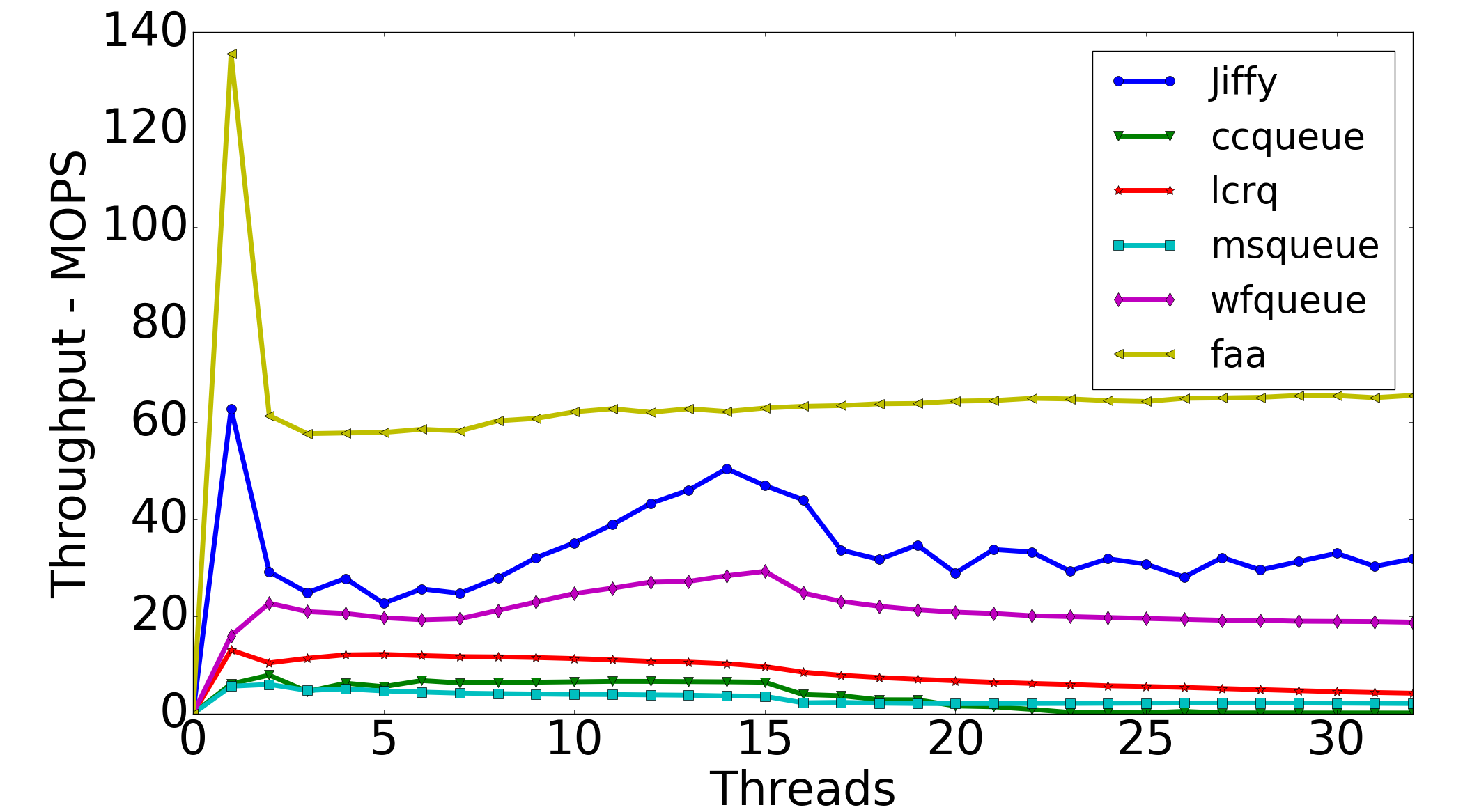}
	}
	\caption{
		Intel Xeon.
	}
	\label{fig:heavyEnqueue}
\end{subfigure}
\caption{Enqueues only - 1 second runs.}
\end{figure}

\begin{figure}[t]
	\begin{subfigure}{0.5\textwidth}
		\center{
		\includegraphics[width=0.93\columnwidth]{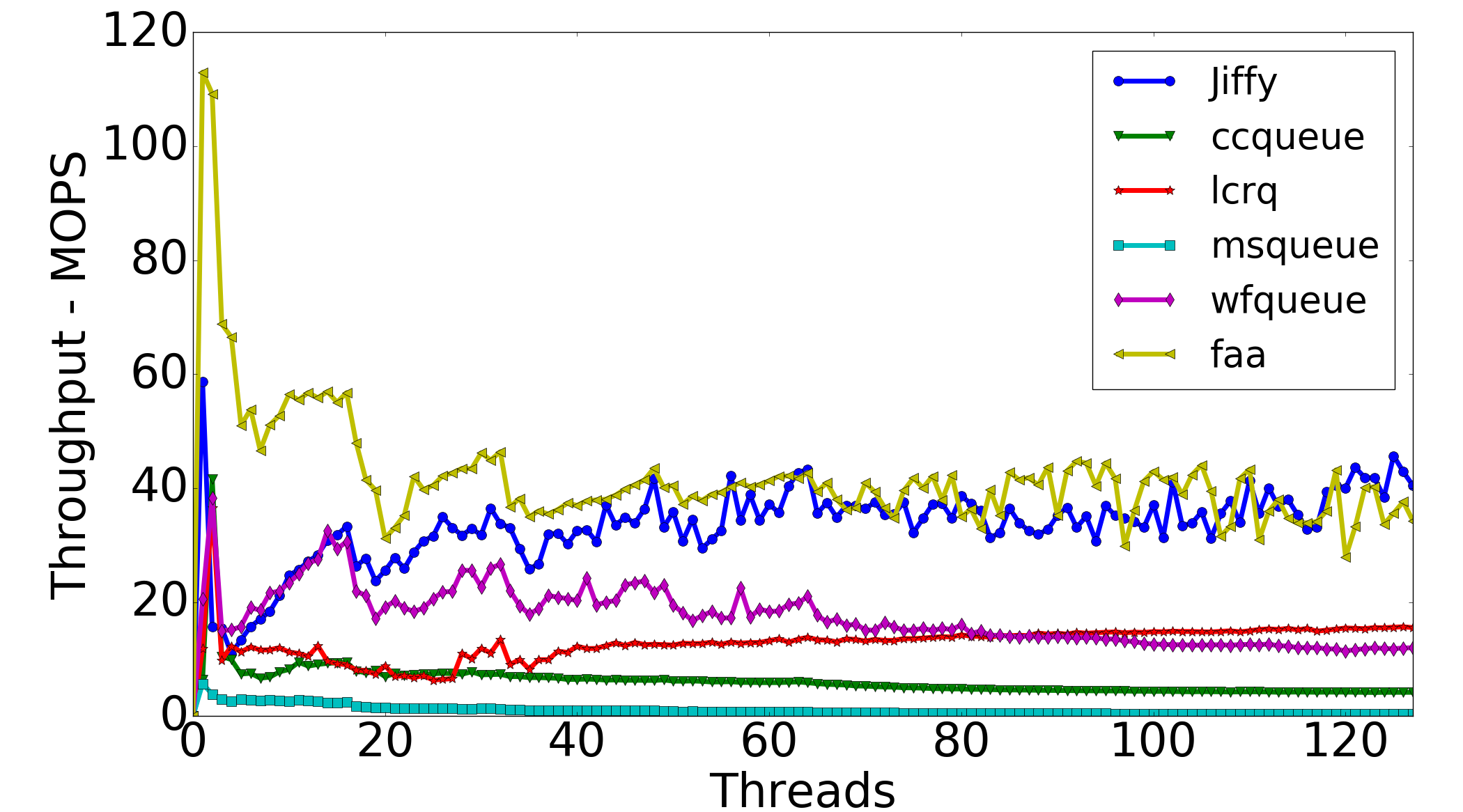}
	}
	\caption{
		AMD PowerEdge.
	}
	\label{fig:limondequeue}	
\end{subfigure}
\begin{subfigure}{0.5\textwidth}
	\center{
		\includegraphics[width=0.93\columnwidth]{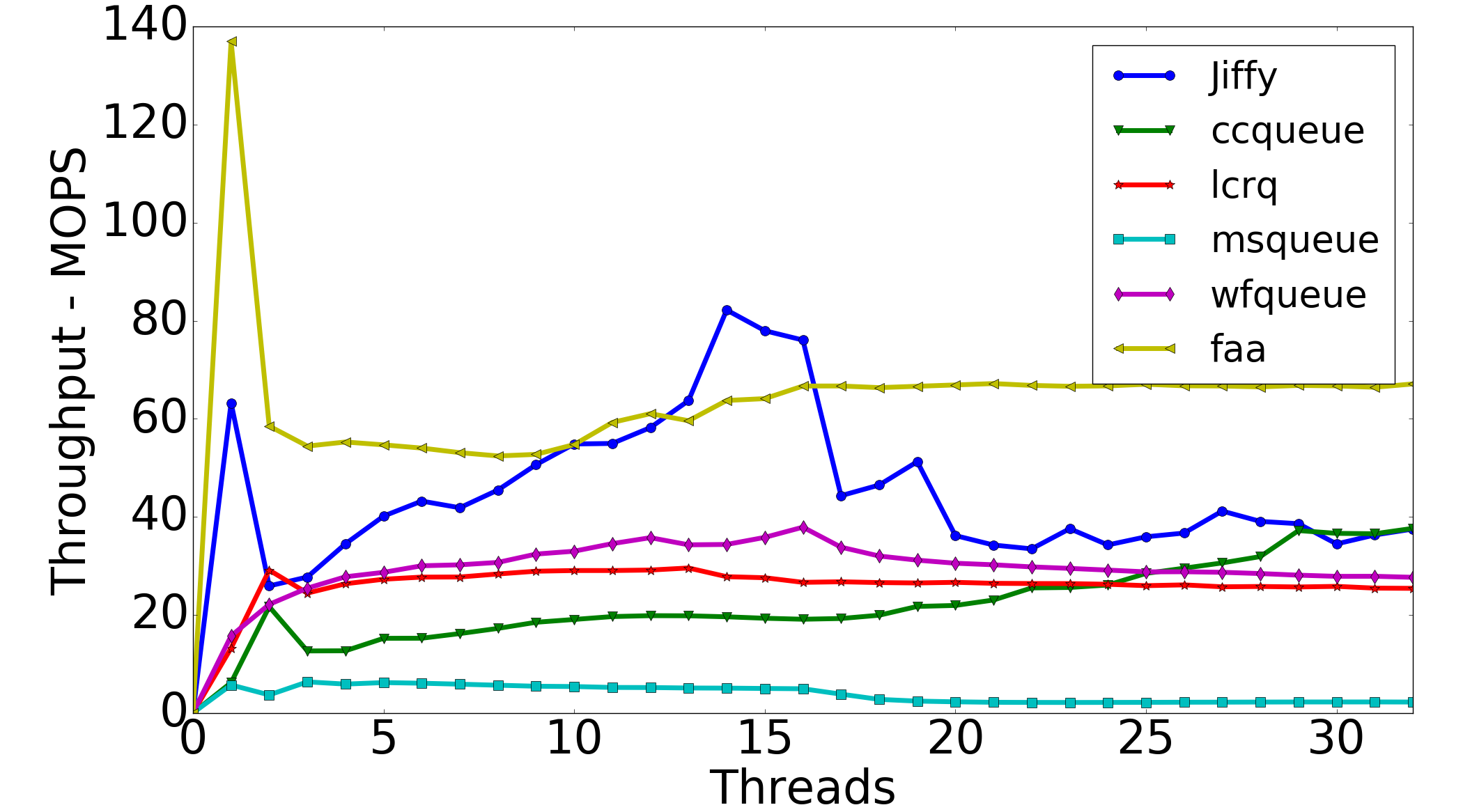}
	}
	\caption{
		Intel Xeon.
	}
	\label{fig:havydequeue}	
\end{subfigure}
\caption{Multiple enqueuers with a single dequeuer - 1 second runs.}
\end{figure}

\begin{figure}[t]
	\begin{subfigure}{0.5\textwidth}
			\center{
		\includegraphics[width=0.93\columnwidth]{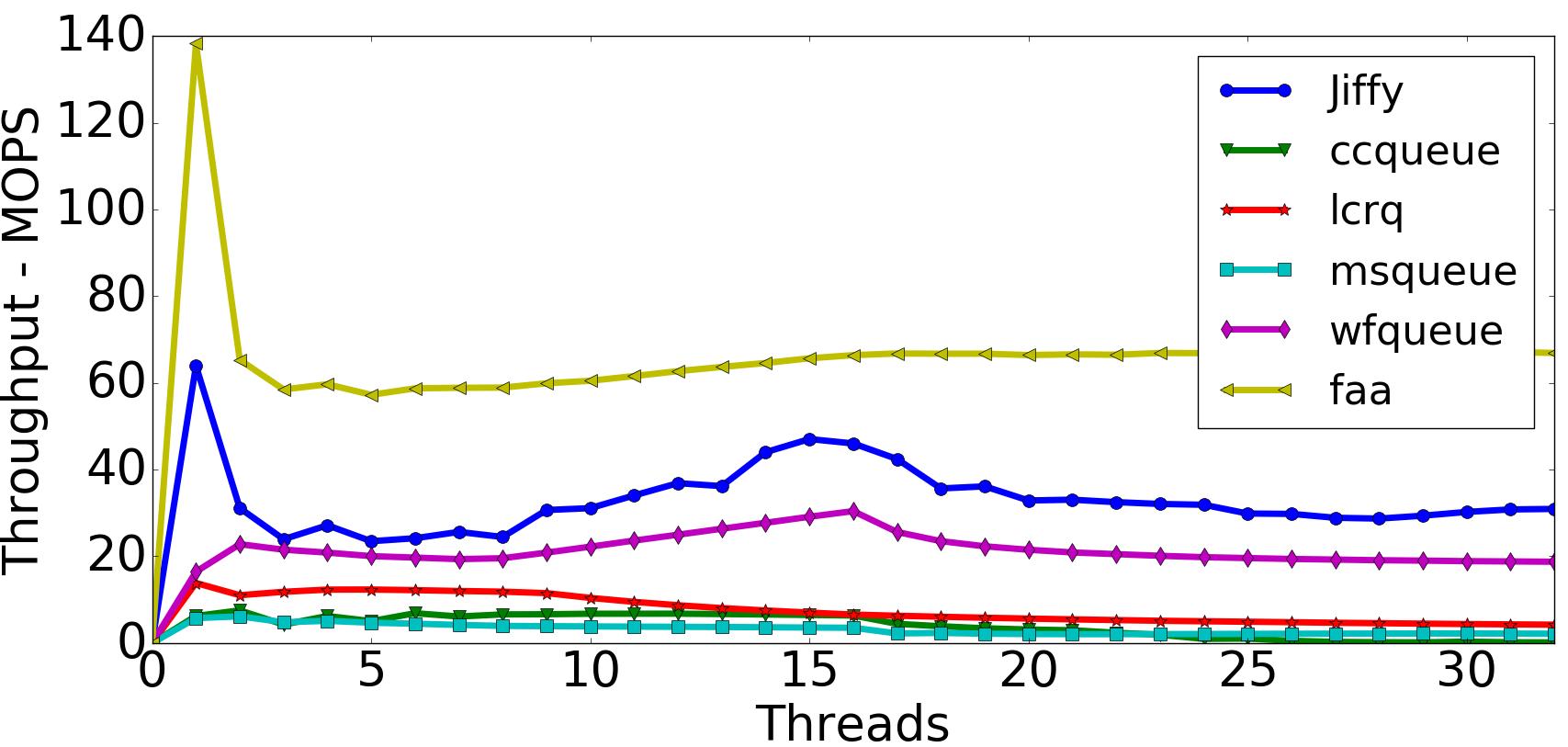}
	}
	\caption{
		Enqueues only.
	}
	\label{fig:10havyenqueue}
\end{subfigure}
\begin{subfigure}{0.5\textwidth}
	\center{
		\includegraphics[width=0.93\columnwidth]{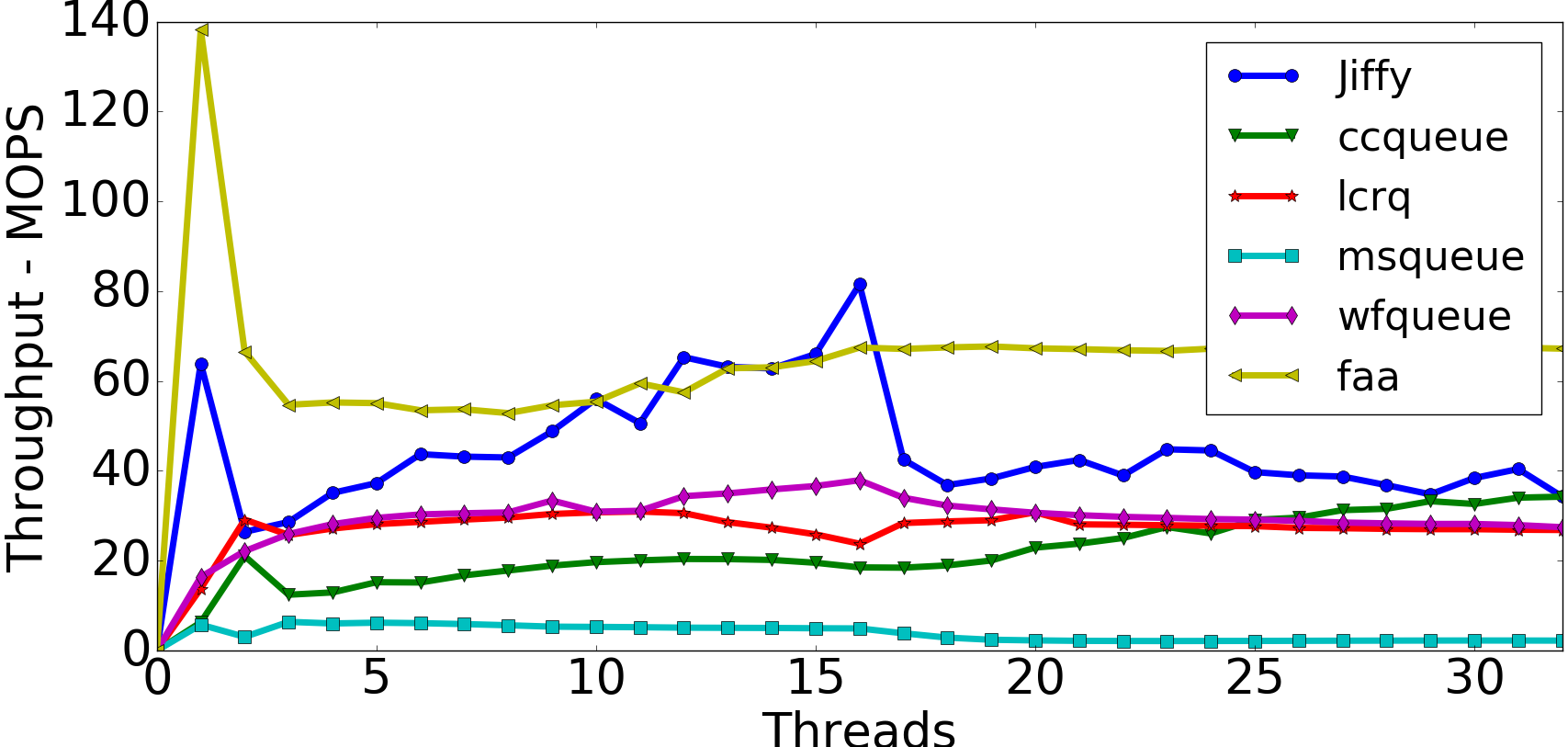}
	}
	\caption{
		Multiple enqueuers with a single deqeuer.
	}
	\label{fig:10havydequeue}
\end{subfigure}
\caption{Intel Xeon - 10 seconds runs.}
\end{figure}

\paragraph*{Total Space Usage}
We collected memory usage statistics via valgrind version 3.13.0~\cite{nethercote2003valgrind}.
We measure the memory usage when inserting $10^7$ elements to the queue on AMD PowerEdge.
Table~\ref{tab:2threads} lists the memory usage when the queues are used by one enqueuer and one dequeuer.
\MYQUEUE's memory usage is significantly smaller than all other queues compared in this work.
\MYQUEUE{} uses 38.7 MB of heap memory which is 93.67\% less than the WFqueue consumption of 611.63 MB, and 97\% less than MSqueue.
The peak heap size depicts the point where memory consumption was greatest.
As can be seen \MYQUEUE's peak is significantly lower than the rest.
\MYQUEUE's miss ratios in L1 and L3 data caches are better due to its construction as a linked list of arrays.

\ifdefined\COVER
Table~\ref{tab:128threads} in Appendix~\ref{sec:additional-results} depicts the memory usage with 127 enqueuers and a single dequeuer.
\MYQUEUE's heap consumption has grown to 77.10 MB, yet it is 87\% less than WFqueue's consumption, 87\% less than CCqueue and
96.8\% less than MSqueue.
As mentioned above, this memory frugality is an important factor in \MYQUEUE's cache friendliness.
\else
\paragraph*{Memory Usage with 128 Threads}

Table~\ref{tab:128threads} lists the Valgrind statistics for the run with 127 enqueuers and one dequeuer.
Here, \MYQUEUE's heap consumption has grown to 77.10 MB, yet it is 87\% less than WFqueue's consumption, 87\% less than CCqueue and
96.8\% less than MSqueue.
As mentioned above, this memory frugality is an important factor in \MYQUEUE's cache friendliness, as can be seen by the cache miss statistics of the CPU data caches (D1 and L3d miss).

\begin{table}[H]
	\resizebox{\textwidth}{!}{%
		\begin{tabular}{|l||l||l|l||l|l||l|l||l|l|}
			\hline
			\multirow{10}{*}{} &
			\multicolumn{1}{c||}{Jiffy} &
			\multicolumn{2}{c||}{WF} &
			\multicolumn{2}{c||}{LCRQ}&
			\multicolumn{2}{c||}{CC}&
			\multicolumn{2}{c|}{MS}\\
			& Absolute  & Absolute &Relative & Absolute &Relative & Absolute & Relative & Absolute & Relative \\
			\hline
			Total Heap Usage & 77.10 MB &611.70 MB & x7.93 &1.19 GB &x15.8 & 605.68 MB & x7.85 &2.36 GB &x31.42 \\
			\hline
			Number of Allocs & 6409 & 10045 & x1.57&2819 &x0.44 & 9922394 &x1548 & 9922263 & x1548 \\
			\hline
			Peak Heap Size & 87.42 MB  &  624.5 MB & x7.14 &1.191 GB  &x13.94 &  796.7 MB &x9.11 &2.426 GB &x28.42 \\
			\hline
			\# of Instructions Executed & 678,651,794 & 3,099,458,758 &x4.57  &1,671,748,656  &x0.99 & 8,909,842,602 &x13.13 & 11,944,994,600 & x17.60  \\
			\hline
			I1  Misses & 2,238 &1,724 & x0.77 &1,669 &x0.75 & 1,668 & x0.75& 1,689 & x0.75 \\
			\hline
			L3i Misses &2,194 & 1,717 &x0.78 & 1,658&x0.76 & 1,664 & x0.76 &1,684 &x0.77 \\
			\hline
			Data Cache Tries (R+W)& 352,980,193 & 1,849,850,595 & x5.24  & 690,008,884 &x1.95 & 3,172,272,116 &x8.99 & 3,237,610,091 & x9.17 \\
			\hline
			D1  Misses & 2,643,266 & 51,230,800 &x19.38 &30,012,317 &x11.35  &68,349,604 &x25.86 & 79,097,745 & x29.92 \\
			\hline
			L3d Misses & 1,298,646 & 40,453,921 & x31.15  & 19,946,542  &x15.36  & 57,287,592 & x44.11 & 64,219,733 & x49.45 \\
			\hline
		\end{tabular}
	}
	\normalsize
	\caption{
		Valgrind memory usage statistics with 127 enqueuers and one dequeuer. I1/D1 is the L1 instruction/data cache respectively while L3i/L3d is the L3 instruction/data cache respectively.
	}
	\label{tab:128threads}
\end{table}

\fi

\paragraph*{Throughput Results}


Figure~\ref{fig:limonEnqueue} shows results for the enqueues only benchmark on the AMD PowerEdge server for one second runs.
\MYQUEUE{} obtains the highest throughput with two threads, as the two threads are running on the same core and different hardware threads. 
The rest of the queues obtain the highest throughput with only one thread.
\MYQUEUE{} outperforms all queues, reaching as much as 69 millions operations per second (MOPS) with two threads.
The FAA is an upper-bound for all the queues as they all preform FAA in each enqueue operation.
The peak at 16 threads and the drop in 17 threads in \MYQUEUE{} is due to the transition into a new NUMA node.
At this point, a single thread is running by itself on a NUMA node.
Notice also the minor peaks when adding a new core vs. starting the 2$^{\mathrm{nd}}$ thread on the same core.

Beyond 64 threads the application already spans two CPUs.
The performance reaches steady-state.
This is because with 2 CPUs, the sharing is at the memory level, which imposes a non-negligible overhead.
\MYQUEUE{} maintains its ballpark performance even when 128 enqueuers are running with a throughput of 22 MOPS.  

Figure~\ref{fig:heavyEnqueue} shows results for the enqueues only benchmark on the Intel Xeon server for one second runs.
All of the queues suffer when there are more threads than hardware threads, which happens beyond 16 threads. 

Figure~\ref{fig:limondequeue} shows results with a single dequeuer and multiple enqueuers on the AMD PowerEdge server for one second runs.
\MYQUEUE{} is the only queue whose throughput improves in the entire range of the graph (after the initial 2-threads drop). 

Figure~\ref{fig:havydequeue} shows results with a single dequeuer and multiple enqueuers on the Intel Xeon server for one second runs.
Here \MYQUEUE{} outperforms in some points even the FAA benchmark.
This is because the dequeuer of \MYQUEUE{} does not preform any synchronization operations.
None of the other queues can achieve this due to their need to support multiple dequeuers.

Figure~\ref{fig:10havyenqueue} shows results for the enqueues only benchmark on the Intel Xeon server when the run length is set to $10$ seconds.
%
Figure~\ref{fig:10havydequeue} shows results for a single dequeuer and multiple enqueuers on the Intel Xeon server with $10$ seconds runs.
As can be seen, in both cases the results are very similar to the $1$ second runs. 
The same holds for the AMD PowerEdge server.

Let us comment that in a production system, the enqueque rate cannot surpass the dequeue rate for long periods of time; otherwise the respective queue would grow arbitrarily.
However, it is likely to have short bursts lasting a few seconds each, where a single shard gets a disproportionate number of enqueues.
This test exemplifies \MYQUEUE's superior ability to overcome such periods of imbalance.

\section{Conclusions}
\label{sec:discussion}

In this paper we presented \MYQUEUE, a fast memory efficient wait-free multi-producers single-consumer FIFO queue.
\MYQUEUE{} is based on maintaining a linked list of buffers, which enables it to be both memory frugal and unbounded.
Most enqueue and dequeue invocations in \MYQUEUE{} complete by performing only a few atomic operations.

Further, the buffer allocation scheme of \MYQUEUE{} is designed to reduce contention and memory bloat.
Reducing the memory footprint of the queue means better inclusion in hardware caches and reduced resources impact on applications. 

\MYQUEUE{} outperforms prior queues in all concurrency levels especially when the dequeuer is present.
%
Moreover, \MYQUEUE's measured memory usage is significantly smaller than all other queues tested in this work, $\approx90\%$ lower than WFqueue, LCRQ, CCqueue, and MSqueue.




\scriptsize
{ \bibliographystyle{abbrv}
	\bibliography{references}
}
\newpage
\appendix
\ifdefined\COVER
\section{Missing Proofs}
\subsection{Proof of Lemma~\ref{lemma-enqueue}}

\subsection{Proof of Lemma~\ref{lemma-dequeue}}

\fi
\section{Detailed Pseudocode for \MYQUEUE{}}
\label{sec:detailedcode}
\paragraph*{Enqueue Operation}
Algorithm~\ref{alg:enqueue} lists the detailed implementation of the enqueue operation, expanding the high level pseudo-code given in Algorithm~\ref{alg:enqueue_high}.
Specifically, if the location index obtained through the fetch\_add(1) operation in line~\ref{alg:FAA} is beyond the last allocated buffer, the enqueuer allocates a new buffer and tries to add it to the queue using a \emph{CAS} operation.
This is performed by the while loop at line~\ref{alg:next} of Algorithm~\ref{alg:enqueue}, which corresponds to the while loop in line~\ref{alg:Hnext} of Algorithm~\ref{alg:enqueue_high}.
If the CAS fails, it means that another thread succeeded, so we can move to the new allocated buffer and check if the location we fetched at line~\ref{alg:FAA} is in that buffer.
If it is not, then we continue trying to add buffers until reaching the correct buffer.

The while loop at line~\ref{alg:back} of Algorithm~\ref{alg:enqueue} matches the loop at line~\ref{alg:Hback} in Algorithm~\ref{alg:enqueue_high}, where if needed, the enqueuer retracts from the end of the queue to the previous buffer corresponding to location.
Line~\ref{line:adjust} in Algorithm~\ref{alg:enqueue} is where we calculate the correct index.
To do that we remove from location, which is a global index for the entire queue, the amount of items up to the current buffer.
This is done at line~\ref{alg:prevSize} of Algorithm~\ref{alg:enqueue}.
We then insert the node in line~\ref{alg:data}.
Yet, just before returning, if the thread is in the last buffer of the queue and it obtained the second index in that buffer, then the enqueuer tries to add a new buffer to the end of the queue at line~\ref{alg:indexone}.

\begin{algorithm}[h]
	\caption{Enqueue operation - detailed}
	\label{alg:enqueue}
	\scriptsize
	\begin{algorithmic}[1]
		\Function {enqueue}{data}
		\State {unsigned int location = tail.fetch\_add(1)}\label{alg:FAA}
		\State{bool isLastBuffer = true}

		\State{bufferList* tempTail = tailOfQueue.load()}\label{alg:tempTail}
		
		\State{unsigned int numElements = bufferSize*tempTail  $\rightarrow$ positionInQueue}  	
	\While{location $\geq$ numElements}\label{alg:next}
		\State  //location is in the next buffer
		\If{(tempTail  $\rightarrow$ next).load() == NULL} 	
		\State  //buffer not yet exist in the queue
		\State{bufferList* newArr = new bufferList(bufferSize, tempTail$\rightarrow$ positionInQueue + 1, tempTail)}		
		\If{CAS(\&(tempTail$\rightarrow$ next), NULL, newArr)}\label{alg:succeed}					
		\State{CAS(\&tailOfQueue,\&tempTail, newArr)}\label{alg:succeed2}		
		\Else	
		\State{delete  newArr}\label{alg:delete}
		\EndIf	
		\EndIf
			\State{tempTail = tailOfQueue.load()} 
		\State{numElements = bufferSize*tempTail  $\rightarrow$ positionInQueue}  
	\EndWhile
					
		\State  //calculating the amount of item in the queue - the current buffer
		\State unsigned  int  prevSize= bufferSize*(tempTail $\rightarrow$ positionInQueue-1) \label{alg:prevSize}
		\While{location < prevSize}\label{alg:back}
			\State  // location is in a previous buffer from the buffer pointed by tail
			\State{tempTail = tempTail$\rightarrow$ prev}
			\State {prevSize = bufferSize*(tempTail $\rightarrow$ positionInQueue - 1)}
			\State{isLastBuffer = false}
		\EndWhile
		
		\State  // location is in this buffer	

		\State{Node* n = \&(tempTail$\rightarrow$ currbuffer[location - prevSize])} \label{line:adjust}
		\If{n$\rightarrow$ isSet.load() ==  State.empty }	
			\State{n $\rightarrow$ data = data}\label{alg:data}
			\State{n $\rightarrow$ isSet.store(State.set)}\label{alg:dataSet}
			
			\If{index == 1 \&\& isLastBuffer}\label{alg:indexone}
				\State  //allocating a new buffer and adding it to the queue
				\State{bufferList* newArr = new bufferList(bufferSize, tempTail$\rightarrow$ positionInQueue + 1, tempTail)}		
				\If{!CAS(\&(tempTail$\rightarrow$ next), NULL, newArr)}	
					\State{delete  newArr}
				\EndIf
			\EndIf

		\EndIf
		
		\EndFunction
	\end{algorithmic}
\end{algorithm}

\paragraph*{Dequeue Operation}
Algorithm~\ref{alg:dequeue} provides the detailed implementation for the dequeue operation, corresponding to the high level pseudo-code in Algorithm~\ref{alg:dequeue_high}.
First, we skip \texttt{handled} elements in lines~\ref{alg:handledStart}--\ref{alg:handledEnd} of Algorithm~\ref{alg:dequeue}, which match the loop at line~\ref{alg:HHandled} of Algorithm~\ref{alg:dequeue_high}.
Next, we check whether the queue is empty and if so return false (lines~\ref{alg:empty}--\ref{alg:CheckEmptyEnd} of Algorithm~\ref{alg:dequeue}).
To do so, we compare the \emph{headOfQueue} and \emph{tailOfQueue} to check if they point to the same \emph{BufferList} as well as compare the tail and the head.
Note that the tail index is global to the queue while the head index is a member of the \emph{BufferList} class.

If the element pointed by head is marked \texttt{set}, we simply remove it from the queue and return (lines~\ref{alg:nSetStart}--\ref{alg:nSetEnd}).
Next, if the first item is in the middle of an enqueue process (\textit{isSet}=\texttt{empty}), then the consumer scans the queue for a later \texttt{set} item.
This is performed by function \textsc{Scan} listed in Algorithm~\ref{alg:scan}, which matches line~\ref{alg:HforN} in Algorithm~\ref{alg:dequeue_high}.
If such an item is found, it is marked \textit{tempN}.
The folding of the queue invoked in line~\ref{alg:Hfold} of Algorithm~\ref{alg:dequeue_high} is encapsulated in the \textsc{Fold} function listed in Algorithm~\ref{alg:fold}.

When preforming the fold, only the array is deleted, which consumes the most memory.
The reason for not deleting the the entire array's structure is to let enqueuers, who kept a pointer for \emph{tailOfQueue} at line~\ref{alg:tempTail} of Algorithm~\ref{alg:enqueue}, point to a valid \emph{BufferList} at all times with the correct \emph{prev} and \emph{next} pointers.
To delete the rest of the array structure later on, the dequeuer keeps this buffer in a list called \emph{garbageList}, a member of the \MYQUEUE{} class (line~\ref{alg:garbage}). 
When a \emph{BufferList} from the queue is deleted, the dequeuer checks if there is a \emph{BufferList} in \emph{garbageList} that is before the buffer being deleted.
If so, the dequeuer deletes it as well (lines~\ref{alg:deleteGarbage}-~\ref{alg:deleteGarbageEnd}).

Before dequeuing, the consumer scans the path from head to \textit{tempN} to look for any item that might have changed its status to \texttt{set}.
This is executed by function \textsc{Rescan} in Algorithm~\ref{alg:rescan}, corresponding to line~\ref{alg:Hscan} of Algorithm~\ref{alg:dequeue_high}.
If such an item is found, it becomes the new \textit{tempN} and the scan is restarted.
This is to check whether there is an even closer item to \texttt{n} that changed its status to \texttt{set}.
Finally, a dequeued item is marked \texttt{handled} in line~\ref{alg:removeTempN} of Algorithm~\ref{alg:dequeue}, matching line~\ref{alg:HremoveN} of Algorithm~\ref{alg:dequeue_high}.
Also, if the dequeued item was the last non-\texttt{handled} in its buffer, the consumer deletes the buffer and moves the head to the next one.
This is performed in line~\ref{alg:advanceHead} of Algorithm~\ref{alg:dequeue}, corresponding to line~\ref{alg:HadvanceHead} of Algorithm~\ref{alg:dequeue_high}. 

\begin{algorithm}[t]
	\caption{Dequeue operation - detailed}
	\label{alg:dequeue}	
	\scriptsize
	\begin{algorithmic}[1]
		\Function {dequeue}{T\&data}
			\State{Node* n = \&(headOfQueue $\rightarrow$ currbuffer[headOfQueue  $\rightarrow$ head]);}
			\While{n $\rightarrow$ isSet.load() == State.handled}\label{alg:handledStart} // find first non-handled item
				\State{headOfQueue$\rightarrow$head++}
				\State{bool res = \textsc{MoveToNextBuffer}()} // if at end of buffer, skip to next one
				\If {!res}
					\State{return false;} // reached end of queue and it is empty
				\EndIf
			\State{n = \&(headOfQueue$\rightarrow$currbuffer[headOfQueue$\rightarrow$head])} // n points to the beginning of the queue	
			\EndWhile \label{alg:handledEnd}

			\State{// check if the queue is empty}	
		\If{((headOfQueue == tailOfQueue.load()) \&\& (headOfQueue$\rightarrow$head == tail.load() \% bufferSize ))}  \label{alg:empty}
		\State{return false}
		\EndIf\label{alg:CheckEmptyEnd}

		\If{n $\rightarrow$ isSet.load() == State.set}  \label{alg:nSetStart} // if the first element is \texttt{set}, dequeue and return it
			\State{headOfQueue$\rightarrow$head++}
			\State{\textsc{MoveToNextBuffer()}}
			\State{data = n$\rightarrow$data}
			\State{return true}	
			\EndIf\label{alg:nSetEnd}

		\If{n $\rightarrow$ isSet.load() == State.empty} // otherwise, scan and search for a \texttt{set} element
		\State{bufferList* tempHeadOfQueue = headOfQueue}
		\State{unsigned int tempHead = headOfQueue$\rightarrow$head}
		\State{Node* tempN = \&(tempHeadOfQueue $\rightarrow$ currbuffer[tempHead])}
			\State{bool res = \textsc{Scan}(tempHeadOfQueue ,tempHead ,tempN)}
			\If {!res}
				\State{return false;} // if none was found, we return empty		
			\EndIf
		
			\State{//here tempN == set (if we reached the end of the queue we already returned false)}	
			
			\State{\textsc{Rescan}(headOfQueue ,tempHeadOfQueue ,tempHead ,tempN)} // perform the rescan

			\State {// \emph{tempN} now points to the first \texttt{set} element -- remove \emph{tempN}}
			\State{data = tempN$\rightarrow$data}
			\State{tempN$\rightarrow$isSet.store(State.handled)}\label{alg:removeTempN}
			\If{(tempHeadOfQueue==headOfQueue  \&\& tempHead ==head) //tempN ==n}
			\State{headOfQueue$\rightarrow$head++}
			\State{\textsc{MoveToNextBuffer()}}\label{alg:advanceHead}
			\EndIf
			\State{return true}
		\EndIf
		\EndFunction
		\algstore*{dequeueBreak}
\end{algorithmic}
\end{algorithm}

\begin{algorithm}[t]
	\caption{Folding a fully \texttt{handled} buffer in the middle of the queue}
	\label{alg:fold}
	\scriptsize
	\begin{algorithmic}[1]
		\algrestore*{dequeueBreak}
		\Function {fold}{bufferList* tempHeadOfQueue, unsigned int\& tempHead, bool\& flag\_moveToNewBuffer, bool\&  flag\_bufferAllHandeld}
		\If{tempHeadOfQueue == tailOfQueue.load()}
		\State{return false // the queue is empty -- we reached the tail of the queue}
		\EndIf
		
		\State{bufferList* next = tempHeadOfQueue$\rightarrow$next.load()}
		\State{bufferList* prev = tempHeadOfQueue$\rightarrow$prev}
		\If{next == NULL}
		\State{return false}  // we do not have where to move	
		\EndIf
		
		\State{// shortcut this buffer and delete it}
		\State{next$\rightarrow$prev = prev}
		\State{prev$\rightarrow$next.store(next)}
			\State{delete[] tempHeadOfQueue$\rightarrow$currbuffer}\label{alg:deleteArr}
			\State{garbageList.addLast(tempHeadOfQueue) }\label{alg:garbage}
		\State{tempHeadOfQueue = next}
		\State{tempHead = tempHeadOfQueue$\rightarrow$head}
		\State{flag\_bufferAllHandeld = true}
		\State{flag\_moveToNewBuffer = true}
		
		\State{return true}	
		
		\EndFunction
		\algstore*{dequeueBreak2}			
	\end{algorithmic}	
\end{algorithm}

\begin{algorithm}[t]
	\caption{MoveToNextBuffer -- a helper function to advance to the next buffer}
	\label{alg:moveToNextBuffer}
	\scriptsize
	\begin{algorithmic}[1]
		\algrestore*{dequeueBreak2}
		\Function {moveToNextBuffer}{}
		\If {headOfQueue$\rightarrow$head $\geq$ bufferSize}
		
		\If {headOfQueue == tailOfQueue.load()}
		\State{return false}			
		\EndIf
		\State{	bufferList* next = headOfQueue$\rightarrow$next.load()}
		
		\If {next == NULL}
		\State{return false}			
		\EndIf
		\State{bufferList* g =garbageList.getFirst() }\label{alg:deleteGarbage}
		\While {g$\rightarrow$positionInQueue < next$\rightarrow$positionInQueue}
		\State{garbageList.popFirst() }
			\State{delete g}
				\State{ g =garbageList.getFirst()}		
		\EndWhile\label{alg:deleteGarbageEnd}
		
		\State{delete headOfQueue}
		\State{headOfQueue = next}
		
		\EndIf
		\State{return true}			
		
		\EndFunction	
		\algstore*{dequeueBreak3}
	\end{algorithmic}	
\end{algorithm}

\begin{algorithm}[t]
	\caption{Scan the queue from n searching for a \texttt{set} element -- return false on failure}
	\label{alg:scan}
	\scriptsize
	\begin{algorithmic}[1]
		\algrestore*{dequeueBreak3}
		\Function {scan}{bufferList* tempHeadOfQueue, unsigned int\& tempHead, Node* tempN}
		
		\State{bool flag\_moveToNewBuffer = false , flag\_bufferAllHandeld=true}		
		
		\While{tempN$\rightarrow$isSet.load() != State.set}
		\State{tempHead++}
		
		\If{tempN$\rightarrow$isSet.load() != State.handled}
		\State{flag\_bufferAllHandeld = false}
		\EndIf
		
		\If{tempHead $\geq$ bufferSize} // we reach the end of the buffer -- move to the next
	
		\If{flag\_bufferAllHandeld \&\& flag\_moveToNewBuffer} // fold fully \textbf{handled} buffers
		\State{bool res = \textsc{Fold}(tempHeadOfQueue, tempHead, flag\_moveToNewBuffer,flag\_bufferAllHandeld)} 
		\If {!res}
			\State{return false;}			
		\EndIf
		\Else 	\State{// there is an empty element in the buffer, so we can't delete it; move to the next buffer}	
		\State{bufferList* next = tempHeadOfQueue$\rightarrow$next.load()}
		\If{next == NULL}
		\State{return false // we do not have where to move}
		\EndIf
		\State{tempHeadOfQueue = next}
		\State{tempHead = tempHeadOfQueue$\rightarrow$head}
		\State{flag\_bufferAllHandeld = true}
		\State{flag\_moveToNewBuffer = true}						
		\EndIf	
		
		\EndIf			
		\EndWhile		
		\EndFunction
		\algstore*{dequeueBreak4}			
	\end{algorithmic}	
\end{algorithm}

\begin{algorithm}[t]
	\caption{Rescan to find an element between \emph{n} and \emph{tempN} that changed from \texttt{empty} to \texttt{set}}
	\label{alg:rescan}
	\scriptsize
	\begin{algorithmic}[1]
		\algrestore*{dequeueBreak4}
		\Function {rescan}{bufferList* headOfQueue ,bufferList* tempHeadOfQueue, unsigned int\& tempHead ,Node* tempN  }		
			\State{// we need to scan until one place before tempN }
			\State{bufferList* scanHeadOfQueue = headOfQueue}
			\For{(unsigned int scanHead = scanHeadOfQueue$\rightarrow$head; ( scanHeadOfQueue != tempHeadOfQueue || scanHead < (tempHead-1) ); scanHead++)}			
			\If{scanHead $\geq$ bufferSize} // at the end of a buffer, skip to the next one
			\State{scanHeadOfQueue= scanHeadOfQueue$\rightarrow$next.load()}
			\State{scanHead = scanHeadOfQueue$\rightarrow$head}	
			\EndIf	
			
			\State{Node* scanN = \&(scanHeadOfQueue $\rightarrow$ currbuffer[scanHead]);}
			\State{// there is a closer element to \emph{n} that is \texttt{set} -- mark it and restart the loop from n}
			\If{scanN $\rightarrow$ isSet.load() == State.set} 
			\State{tempHead = scanHead}
			\State{tempHeadOfQueue = scanHeadOfQueue}	
			\State{tempN = scanN}
			\State{scanHeadOfQueue = headOfQueue}
			\State{scanHead = scanHeadOfQueue$\rightarrow$head}	
			
			\EndIf	
			
			\EndFor
		\EndFunction		
	\end{algorithmic}	
\end{algorithm}

\ifdefined\COVER
\section{Additional Figures and Performance Results}
\label{sec:additional-results}

\paragraph*{Queue Types}
Figure~\ref{fig:queues} exemplifies the difference between an SMPC queue and an MPSC queue.
SPMC queues are useful in master worker architectures.
In such scenarios, a single master queues up tasks to be executed, which are picked by worker threads on a first comes first served basis.
In contrast, MPSC queues are often helpful in sharded architectures.
In these cases, each shard is served by a single worker thread to avoid synchronization inside the shard. Multiple collector threads can feed the queue of each shard.

\begin{figure}[H]
\subfigure[SPMC queue used in a master worker architecture.]{
		\includegraphics[width=0.25\textwidth]{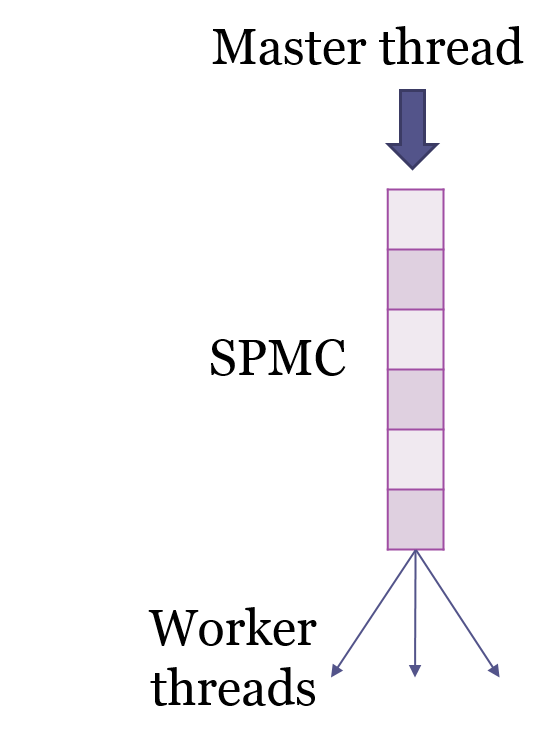}
		\label{fig:SPMC}
}
\subfigure[MPSC queues used in a sharded architecture.]{
		\includegraphics[width=0.3\textwidth]{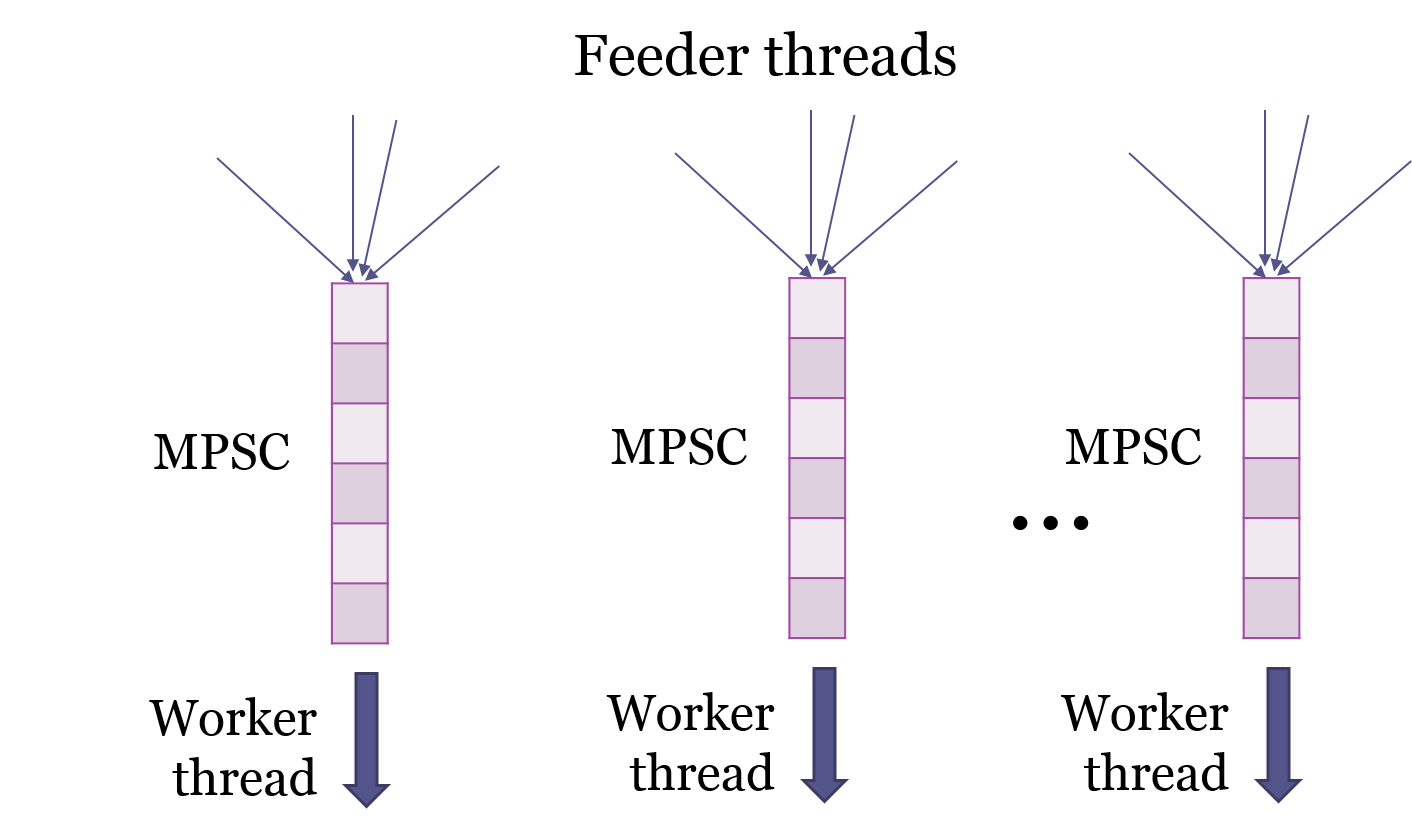}	
		\label{fig:MPSC}
}
	\caption{SPMC vs. MPSC queues.}
	\label{fig:queues}
\end{figure}

\ifdefined\COVER
\paragraph*{The Linearizability Pitfall of the Basic Queue Idea}
Figure~\ref{fig:linearizable} exhibits why the naive proposal for the queue violates linearizability.
This scenario depicts two concurrent enqueue operations $enqueue_1$ and $enqueue_2$ for items $i_1$ and $i_2$ respectively, where $enqueue_1$ terminates before $enqueue_2$.
Also, assume a dequeue operation that overlaps only with $enqueue_2$ (it starts after $enqueue_1$ terminates).
It is possible that item $i_2$ is inserted at an earlier index in the queue than the index of $i_1$ because this is the order by which their respective \emph{FAA} instructions got executed.
Yet, the insertion of $i_1$ terminates quickly while the insertion of $i_2$ (into the earlier index) takes longer.
Now the dequeue operation sees that the \emph{head} of the queue is still empty, so it returns immediately with an empty reply.
But since $enqueue_1$ terminates before the dequeue starts, this violates linearizability.
\begin{figure}[H]
	\center{
		\includegraphics[width = 0.22\textwidth]{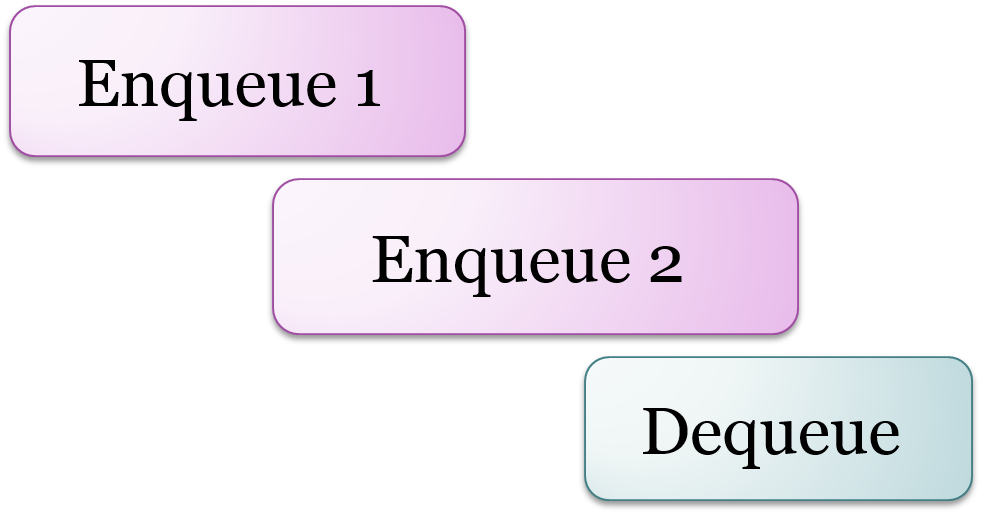}
	}
	\caption{A linearizability problem in the basic queue.}
	\label{fig:linearizable}
\end{figure}
\fi

\ifdefined\COVER
\paragraph*{Memory Buffer Pool Optimization}
Instead of always allocating and releasing buffers from the operating system, we can maintain a buffer pool.
This way, when trying to allocate a buffer, we first check if there is already a buffer available in the buffer pool.
If so, we simply claim it without invoking an OS system call.
Similarly, when releasing a buffer, rather than freeing it with an OS system call, we can insert it into the buffer pool.
It is possible to have a single shared buffer pool for all threads, or let each thread maintain its own buffer pool.
This optimization can potentially reduce execution time at the expense of a somewhat larger memory heap area.
The code used for \MYQUEUE's performance measurements does \emph{not} employ this optimization.
In \MYQUEUE, allocation and freeing of buffers are in any case a relatively rare event.
\fi

\ifdefined\COVER
\paragraph*{Memory Usage with 128 Threads}

Table~\ref{tab:128threads} lists the Valgrind statistics for the run with 127 enqueuers and one dequeuer.
Here, \MYQUEUE's heap consumption has grown to 77.10 MB, yet it is 87\% less than WFqueue's consumption, 87\% less than CCqueue and
96.8\% less than MSqueue.
As mentioned above, this memory frugality is an important factor in \MYQUEUE's cache friendliness, as can be seen by the cache miss statistics of the CPU data caches (D1 and L3d miss).

\begin{table}[H]
	\resizebox{\textwidth}{!}{%
		\begin{tabular}{|l||l||l|l||l|l||l|l||l|l|}
			\hline
			\multirow{10}{*}{} &
			\multicolumn{1}{c||}{Jiffy} &
			\multicolumn{2}{c||}{WF} &
			\multicolumn{2}{c||}{LCRQ}&
			\multicolumn{2}{c||}{CC}&
			\multicolumn{2}{c|}{MS}\\
			& Absolute  & Absolute &Relative & Absolute &Relative & Absolute & Relative & Absolute & Relative \\
			\hline
			Total Heap Usage & 77.10 MB &611.70 MB & x7.93 &1.19 GB &x15.8 & 605.68 MB & x7.85 &2.36 GB &x31.42 \\
			\hline
			Number of Allocs & 6409 & 10045 & x1.57&2819 &x0.44 & 9922394 &x1548 & 9922263 & x1548 \\
			\hline
			Peak Heap Size & 87.42 MB  &  624.5 MB & x7.14 &1.191 GB  &x13.94 &  796.7 MB &x9.11 &2.426 GB &x28.42 \\
			\hline
			\# of Instructions Executed & 678,651,794 & 3,099,458,758 &x4.57  &1,671,748,656  &x0.99 & 8,909,842,602 &x13.13 & 11,944,994,600 & x17.60  \\
			\hline
			I1  Misses & 2,238 &1,724 & x0.77 &1,669 &x0.75 & 1,668 & x0.75& 1,689 & x0.75 \\
			\hline
			L3i Misses &2,194 & 1,717 &x0.78 & 1,658&x0.76 & 1,664 & x0.76 &1,684 &x0.77 \\
			\hline
			Data Cache Tries (R+W)& 352,980,193 & 1,849,850,595 & x5.24  & 690,008,884 &x1.95 & 3,172,272,116 &x8.99 & 3,237,610,091 & x9.17 \\
			\hline
			D1  Misses & 2,643,266 & 51,230,800 &x19.38 &30,012,317 &x11.35  &68,349,604 &x25.86 & 79,097,745 & x29.92 \\
			\hline
			L3d Misses & 1,298,646 & 40,453,921 & x31.15  & 19,946,542  &x15.36  & 57,287,592 & x44.11 & 64,219,733 & x49.45 \\
			\hline
		\end{tabular}
	}
	\normalsize
	\caption{
		Valgrind memory usage statistics with 127 enqueuers and one dequeuer. I1/D1 is the L1 instruction/data cache respectively while L3i/L3d is the L3 instruction/data cache respectively.
	}
	\label{tab:128threads}
\end{table}
\fi

\fi

\end{document}